\documentclass[english,10pt, a4paper]{article}
\usepackage{fullpage}
\usepackage[T1]{fontenc}
\usepackage{babel}
\usepackage{amsmath, amsfonts, amsthm, amssymb}
\usepackage{tikz}
\usepackage{relsize}
\usepackage{graphicx, subfig}

\newcommand{\cS}{\ensuremath{\mathcal{S}}}

\newcommand{\cU}{\ensuremath{\mathcal{U}}}

\newcommand{\R}{\ensuremath{\mathbb{R}}}
\newcommand{\N}{\ensuremath{\mathbb{N}}}

\newcommand{\conv}{*}
\newcommand{\deconv}{\oslash}

\newcommand{\Rmin}{\R_{+, \min}}

\newtheorem{lemma}{Lemma}
\newtheorem{example}{Example}
\newtheorem{theorem}{Theorem}
\newtheorem{defi}{Definition}

\title{Sub-additive service curves in the Network Calculus analysis}
\author{Anne Bouillard\\ENS, France\\{\tt anne.bouillard@ens.fr}}
\date{}
\begin{document}
	\maketitle
	\begin{abstract}
		Network Calculus is a theoretical model that aims at providing upper bounds of worst-case performance (such as delay or buffer occupancy). This is a mathematical framework that handles both network modeling and network analysis. As such it has requirements regarding the space of functions needed for a safe analysis. Namely, the functions  need to be non-negative, as they model a quantity of data. This results in some pitfall for the analysis, where hypothesis matter. 
		
		A recent paper by Hamscher et al. states that allowing functions with negative values can also lead to a valid analysis, in cases that would be untractable with the non-negative assumption results, especially when feedback control is present in the system. 
		
		In this paper, we  show that, on the contrary, a more conventional analysis is possible in all the mentioned cases. The key is a detailed analysis of sub-additive functions.  Second, we show that the analysis of complex feedback control systems, presented by Hamscher et al. in a second paper that uses functions with negative values,  is unsound and has stability issues. We give a corrected analysis, when possible, with conventional hypotheses. 
		
	\end{abstract}
	
	\section{Introduction}
	
	Providing safe, accurate and deterministic performance (buffer occupancy or end-to-end delays) guarantees in network is a strong requirements in modern networks. 
	Adding control at the admission or in parts of the system is key to reach these requirements. Example of applications are Network on Chips~\cite{BGDM20, IS25}, that are based on a window-flow-control protocol~\cite{Ilyas1987}.
	
	Analyzing the worst-case performance of a system requires  compatible analysis tools, different from the more classical stochastic analysis used for queuing systems. Network Calculus is a theory that has been proved adapted in the past decades. It originates from the 90s~\cite{Cruz91, Cruz91a}, and has proved its applicability for the analysis of various stystems: switched networks~\cite{Cruz1995}, ATM networks~\cite{MR2006}, AFDX (Avionics Full Duplex) networks~\cite{BNOT2010} and TSN (Time-sensitive networks)~\cite{BD19}. 
	
	Network Calculus is a theory , introduced by Cruz~\cite{Cruz91}, based on the (min,plus) semi-ring that is adapted to the computation of worst-case performance upper bounds. It is both a model for networks and a collection of analytical tools. 
	
	\paragraph{Network Calculus as a model.} Regarding the model, the main features are the flow or arrival model ({\em arrival curves}), that bounds the quantity of data during each interval of time, and a server model, that represented how data is served, both as the service speed ({\em service curve}) and the service order ({\em scheduling policy}). This model can be extended to network representation, a collection of flows and servers. 
	In addition, window-flow control can be modeled by this framework. In such a system, the entry point admits into the network only a fixed maximum of data that have not been acknowledged yet. It then requires some synchronization and the {\em minimum} operator plays the role of synchronization for the control~\cite{ACOR99}. The analysis has then been extended to the case of one flow is controlled by a series on windows, to ensure no buffer overflow~\cite{BPC09}. 
	
	\paragraph{Network Calculus as an analytical tool.}Regarding the analysis, several techniques have been developed, starting from the often-called {\em fundamental theorem} (Theorem~\ref{th:fund}), that allows computing performance bounds and emphasizes modularity. But multi-flow raises the need for a stronger model of service curve: the {\em strict service curve} in opposition to the former {\em (min,plus) service curve}. That extention has made possible network analysis with various techniques  such as the {Pay-Multiplexing-Only-Once} (PMOO)~\cite{SZM08, BGLT2007} or linear-programming~\cite{BJT10}. In addition, several scheduling policies have been precisely analyzed under this framework: FIFO~\cite{BS12, LMMS2006}, Fixed priorities~\cite{BJT09}, Processor-sharing~\cite{BL18, Bou21} and Round-Robin~\cite{TLB20, Bou23, BSS12}. 
	
	\paragraph{The problem with service curves}
	Nevertheless, this need for a second type of service curve might become an issue, and the risk is that some type of networks may not have a feasible analysis. This is  the issue raised in~\cite{HCS24}, where the authors introduce another extension of the model, by allowing functions with negative values (when the {\em conventional} assumption was to work on non-negative functions only). At first sight, this solves the issue
	(strict service curves are not required anymore)
	at the cost of a) requiring {\em minimal arrival curves} (the arrival must have a minimum arrival rate); b) making the worst-case performance computations slightly more complex and less accurate. 
	
	However, it still remains unclear that the use-cases they advertise, namely sequence of computation/communication systems, and a simple feed-back control,  are useful, since we can present a {\em conventional analysis}. More importantly, Hamscher et al. extend this work to  complex feedback structures. We claim in this paper that this approach is incorrect: as it may raise stability issues. Also, where the approach can be corrected, conventional  techniques can still be used, jeopardizing the need for functions with negative values as service curves.  The key to our analysis is the detailed analysis of sub-additive service curves, and our main result is that non-negative service curves can be computed in a multiple-flow analysis. 

	\paragraph{Contributions:}
	In this report, we analyze the use-cases of~\cite{HCS24} with a conventional analysis and correct~\cite{HWS25} where possible at low cost, complex analysis be beyond the scope of this paper. 
	The key contribution is the study of sub-additive functions as a service curve: we will show that this shape of service curve that has receives very little attention as a service curve has properties similar to strict service curves: non-negative per-flow service curves can be derived. More precisely our contributions are:  

	\begin{itemize}
		\item We exhibit a class, the  sub-additive functions, of service curves where a non-negative per-flow service curve can be computed, even in the case of a (min,plus)-service curve (Section~\ref{sec:sub-add});
		\item We present two cases of application: a slightly generalized PMOO formula in the case of (min,plus) service curves and communication delays (Section~\ref{sec:pmoo}), that requires some comments on the modeling of bounded delay servers ; feedback control (Section ~\ref{sec:wfc}), covering the two cases where ~\cite{HCS24} required service curves with negative values. %We also show that, the approach can also be extended to bounded delay servers, with its usual interpretation. In total, this slightly generalizes the results of \cite{BGLT2007}; 
		\item We demonstrate that the analysis of system with complex feedback loops presented in~\cite{HWS25} is incorrect, and can led to instabilities (Section~\ref{sec:cplx}).  We partially correct the analysis, in simpler cases, using the sub-additive functions. This contribution  is inspired by  the foundation paper~\cite{ACOR99} introducing the concept of {\em throttle service curve}. 
	\end{itemize}
	
	\section{Network Calculus framework}
	In this section, we recall the basic framework of Network Calculus, and its fundamental results. More detailed references about  Network Calculus can be found in \cite{BBL18, Chang2000, LT2001}. 
	
	We denote by $\N = \{0, 1, \ldots\}$ the set of natural integers and for $n\geq 1$, $\N_n = \{1,2,  \ldots, n\}$, and $\R_+$ is the set of non-negative numbers. 
	\subsection{The network model}
	A network is composed of $n$ network elements, also called servers, and $m$ flows of data circulate in this network.
	Throughout the paper, we assume a tandem network: servers are numbered from 1 to $n$ from left to right. 
	 Each flow $i$ of data is associated path $\pi_i$ that us the sequence of servers it crosses, $\pi_i = \langle f_i, f_{i+1}\ldots \ell_i \rangle$ : flow $i$ crosses servers from $f_i$ to $\ell_i$, $f_i\leq \ell_i$.  We denote by $A_i^{(j)}:\R_+\to \R_+$ the {\em cumulative arrival process} of flow $i$ at server $j\in\{f_i,\ldots ,  \ell_i\}$, where for all $t\geq 0$, $A_i^{(j)}(t)$ is the amount of data that arrives at server $j$ between up to time $t$ (more precisely, during the time interval $[0, t)$. Similarly, $D_i^{(j)}$ is the {\em cumulative departure  process} of flow $i$ from server $j$, and $D_i^{(j)}(t)$ is the amount of data of flow $i$ that left server $j$ up to time $t$. Figure~\ref{fig:net} represents one flow in the system with the notations. 
	 
	 \begin{figure}[htbp]
	 	\centering
	 	\begin{tikzpicture}
	 		\node[rectangle, draw, minimum width=0.5cm](s1) at (0, 0) {$1$};
	 		\node[rectangle, draw, minimum width=0.5cm](s2) at (1.5,  0) {$f_i$};
	 		\node[rectangle, draw, minimum width=0.5cm](s3) at (3,  0) {$j$};
	 		\node[rectangle, draw, minimum width=0.5cm](s4) at (4.5, 0) {$j+1$};
	 		\node[rectangle, draw, minimum width=0.5cm](s5) at (6, 0) {$\ell_i$};
	 		\node[rectangle, draw, minimum width=0.5cm](s6) at (7.5, 0) {$n$};
	 		
	 		\draw[black, thick, dotted] (s1) -- (s2) -- (s3)-- (s4)-- (s5)-- (s6);
	 		\draw[black, thick]  (s3)-- (s4);
	 		\draw[red, ultra thick, ->] (1, 0.5)  -- (1, -0.2) node[pos=0, left=-0.4cm] {$A_i^{(f_i)}$}-- (6.7, -0.2)node[pos=0.5, below=] {$A_i^{(j+1)} = D_i^{(j)}$} -- (6.7, 0.5) node[pos=1, left=-0.4cm] {$D_i^{(\ell_i)}$};
	 	\end{tikzpicture}
	 	\caption{Example of a tandem network with some arrival and departure processes of flow $i$. The departure process at server $j$ is the arrival process at server $j+1$: $D_i^{(j)} = A_i^{(j+1)}$. }
	 	\label{fig:net}
	 \end{figure}
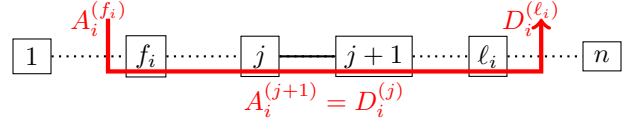 
	  An arrival process $A_i^{(j)}$ is by construction non-negative, non-decreasing, left-continuous, and $A_i^{(j)}(0) = 0$. We also assume causality of all servers: $A_{i}^{(j)}\geq D_i^{(j)}$, for all $i\in\N_k$ and all $j\in\{f_i,\ldots, \ell_i\}$, and $D_i^{(j)}(t)$ depends only on $\big(A_{i}^{(j)}(s)\big)_{s\leq t}$. We also assume that there is no loss or creation of data in the system. 
	
	%In the following, we will denote $i\in j$ or $j\ni i$ if flow $i$ traverses server $j$.
	
	Network calculus is a formal model that characterizes the arrival processes (by the mean of {\em arrival curves}) and the relations between the arrival and the departure processes (by the mean of {\em service curves}), which is the focus of the next paragraph. 
		
	\subsection{Network Calculus model: arrival  and service curves}

%	\begin{figure}[htbp]
%		\centering
%		\begin{tikzpicture}[server/.style={shape=rectangle,draw,minimum height=.8cm,inner xsep=4ex}]
%			\node[server,name=S1] at (0,0) {$\cS$};
%			\draw[->, blue] (-1.2,-0) node[left] {$A$} -- (1.5,0) node[right] {$D$};
%		\end{tikzpicture}
%		\caption{Server model.}
%		\label{fig:server}
%	\end{figure}
%	

	Arrival and service curve are the are the basic elements of the Network Calculus model. They provide  {\em deterministic}, or worst-case characterizations of processes and server elements.  The model is based on the (min,plus) semi-ring, where $+\infty$ is the zero element. We will denote $\R_{+, \min} = \R_+\cup\{+\infty\}$. Unless stated otherwise and throughout the paper, we will consider that all functions are non-decreasing, non-negative and left-continuous. 
	
	\paragraph{Arrival curve:} Let $A$ be a cumulative arrival process (the cumulative process is generic, so we do not put any index or exponent here), and $\alpha:\R_+\to\Rmin$ be a function. We say that $A$ is $\alpha$-constrained  for all $0\leq s\leq t$,  $A(t)-A(s)\leq \alpha(t-s)$. In other words, $\alpha(d)$ is an upper bound on the quantity of data arriving during any interval of duration $d$. 
	
	A classical model for the arrival curve is the token-bucket (or leaky bucket) function: $\gamma_{r, b}:0\mapsto 0;~t>0\mapsto b+rt$, with $b\geq 0$ modeling the maximum amount of data that can arrive simultaneously (the {\em burst}), and $r$ being the {\em arrival rate}, or maximum long-term arrival rate. Without loss of generality, one can assume that an arrival curve is non-negative, non-decreasing, null at 0 and left-continuous. 
	
	\begin{figure}[htbp]
		\centering
		\begin{tikzpicture}
			\draw[->] (0,0) -- (3,0)node[below, pos = 0.9] {time};
			\draw[->] (0,0) -- (0,2) node[left, pos = 0.8] {\rotatebox{90}{data}};
			\draw[red] (0,0) -- (0, 0.25) -- (1, 0.55) -- (1.5,1.5) -- (3,2.5) node[below = -0.05cm, pos = 0.8] {$A$};
			\draw[blue] (0, 0.62) -- (3, 2.62) node[pos=0.7, above] {$\alpha$};
			\draw[blue, dashed] (1, 0.55) -- (1, 1.17) node [pos=0, black, below] {$x$} -- (2.5, 2.17);
			%\draw[blue] (0,0) -- (0.5,0) -- (0.8333, 0.667) -- (1.5, 1) --  (2,1.5) -- (3,1.63) node[below, pos = 0.5] {$D$};
			%\draw[<->, yellow!70!black, thick]   (0.5, 0) node[below]{$t$} -- (0.5, 0.75) node[left, pos = 0.5] {\rotatebox{90}{\tiny $b(t)$}};
			%\draw[green!70!black, <->, thick] (1.2,1.2) -- (1.7,1.2) node[above = -0.1cm, pos = 0.75] {\tiny $d(u)$};
			%\draw[green!70!black, dotted] (1.7,1.2) -- (1.7, 0) node[below = 0.05cm] {$u$};
		\end{tikzpicture}
		\caption{Example of arrival process $A$ and one arrival curve $\alpha$ for $A$. When $\alpha$ is drawn from any point $x$ of $A$, $A$ must remain below $s+\alpha$. }
		\label{fic:ac}
	\end{figure}
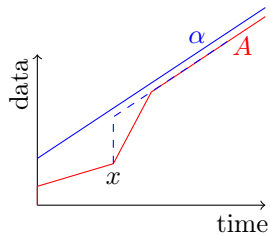
	
	We now model a server. 
	
	\paragraph{(min,plus) service curve:} consider a generic server (no exponent), crossed by the set of flows $I$. For each flow $i\in I$, we denote by $A_i$ and $D_i$, its cumulative arrival and departure processes. Let us define $A= \sum_{i\in I} A_i$ and $D = \sum_{i\in I} D_i$, respectively the aggregate cumulative arrival and departure processes of the server. Let $\beta :\R_+\to \Rmin$ be a non-decreasing and non-negative and left-continuous function such that $\beta(0)= 0$. We say that the server offers the (min,plus) service curve $\beta$ if $D\geq A \conv \beta$, where $\conv$ is the (min, plus)-convolution: $\forall t\geq 0$, $f\conv g(t) = \inf_{0\leq s\leq t} f(s)+f(t-s)$.
	
	\paragraph{Strict service curve:}
	Usually, the notion of (min,plus) service curves is too weak to analyze networks, and a stronger definition has been given: the {\em strict service curve}. We first  define a  backlogged period, as an interval of time $(s, t]$ during which  the server is not empty: $\forall u\in(s, t]$, $A(u)>D(u)$, and then can define the 
	{\em start of backlogged period} of $t$ as the last time before $t$ when the server is empty: $\forall t\geq 0$, $start(t) = \sup\{0\leq s\leq t~|~A(s) = D(s)\}$, with the convention that that when the server is empty at time $t$, $start(t) = t$. We say that the server offers a strict service curve $\beta$ to $A$ if for all $t\geq 0$, $\forall s\in[start(t), t]$, $D(t)-D(s) \geq \beta(t-d)$. Consequently, with $s = start(t)$, we have $$D(t)\geq A(start(t)) + \beta(t-start(t)).$$ 
	%Note that when the server is empty at time $t$, $start(t) = t$, and we obtained the desired property $D(t) = A(t)$. 
	Note that if $\beta$ is a strict service curve for a process, it is also a (min,plus) service curve.

	Classical models for the service curve are 
	\begin{itemize}
		\item The constant-rate service curve $\lambda_R:t\mapsto Rt$, where data is served at least at rate $R$. 
		\item The pure-delay service curve $\delta_d:t\geq d \mapsto0; t>d\mapsto +\infty$ is such that $D(t) \geq A ((t-d)_+)$, hence each bit of data leaves the server after at most a duration $d$. 
		\item The {\em rate-latency service curve}, that combines the two previous curves: $\beta_{R, T} (t):t\mapsto R(t-T)_+$, where $(x)_+ = \max (0, x)$. It can be interpreted as the date entering the server first having to wait a time $T$, before being served at rate $R$.  
	\end{itemize}
	
	Consider $\beta_{R, T}$ a service curve. Figure~\ref{fig:sc} illustrates the difference between a (min, plus) and a strict service curve. 
	
	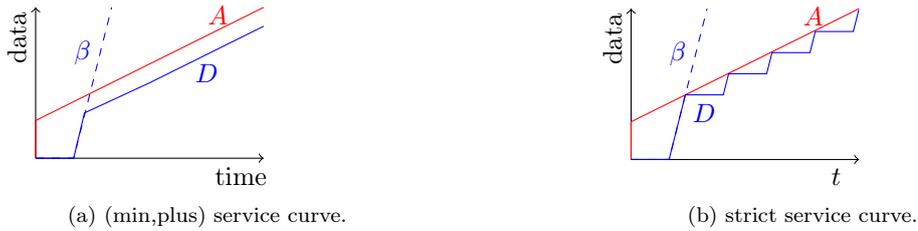
\begin{figure}[htbp]
		\centering
		\subfloat[(min,plus) service curve.]{
			\centering
			\begin{tikzpicture}
				\draw[->] (0,0) -- (3,0)node[below, pos = 0.9] {time};
				\draw[->] (0,0) -- (0,2) node[left, pos = 0.8] {\rotatebox{90}{data}};
				\draw[red] (0,0) -- (0, 0.5) -- (3,2) node[above = -0.05cm, pos = 0.8] {$A$};
				\draw[blue, dashed] (0,0) -- (0.5,0) -- (1, 2) node[pos=0.7, left] {$\beta$};
				\draw[blue] (0,0) -- (0.5,0) -- (0.642, 0.6) -- (1.5, 1) --  (3,1.75)  node[below, pos = 0.5] {$D$};
				\node at (5, 0) {};
		\end{tikzpicture}}
		\hspace{2cm}
		\subfloat[strict service curve.]{
			\centering
			\begin{tikzpicture}
				\draw[->] (0,0) -- (3,0)node[below, pos = 0.9] {$t$};
				\draw[->] (0,0) -- (0,2) node[left, pos = 0.8] {\rotatebox{90}{data}};
				\draw[red] (0,0) -- (0, 0.5) -- (3,2) node[above = -0.05cm, pos = 0.8] {$A$};
				\draw[blue, dashed] (0,0) -- (0.5,0) -- (1, 2) node[pos=0.7, left] {$\beta$};
				
				\draw[blue] (0,0) -- (0.5,0) -- (0.714, 0.857) -- (1.214, 0.857)node[below, pos = 0.5] {$D$} --  (1.284,  1.137) -- (1.784,1.137)  -- (1.854, 1.417) -- (2.354, 1.417)-- (2.424, 1.697) -- (2.924, 1.697) -- (2.994, 1.977) -- (3, 1.977);
				%\draw[green!70!black, <->, thick] (0,1) -- (2,1) node[below = -0.1cm, pos = 0.2] {$d_{\max}$};
				%\draw[yellow!70!black, <->, thick] (1,0) -- (1, 1.25) node[right = -0.1cm, pos = 0.5] {\rotatebox{90}{$b_{\max}$}};
				%\draw[red!70!black, <->, thick] (0,1.65) -- (2.7, 1.65) node[above = -0.1cm, pos = 0.5] {$\ell_{\max}$};
				\node at (5, 0) {};
		\end{tikzpicture}}
		\caption{Illustration of a service curve: (a) a (min, plus) service curve: since there is a non-zero latency, there are cases where there is only one backlogged period. (b) a strict service curve: there is a succession of backlogged periods, and the server is infinitely often empty. }
		\label{fig:sc}
	\end{figure}

	\subsection{Network Calculus analysis: performance and fundamental results}
	One advantage of Network Calculus is that worst-case performance upper bounds can be computed quite easily, only using the arrival and service curve characterization, and not the precise cumulative processes, that may not be entirely defined or known.

	\paragraph{Performance guarantees}
	In this paragraph, we consider a server crossed by a single flow. We denote $A$ its the cumulative arrival process and $D$ its the cumulative departure process . We also assume  that $A$ is $\alpha$-constrained and that the server offers a service curve $\beta$ (of any type).

	We define two types of performance guarantees: 
\begin{enumerate}
		\item[(a)] the maximum backlog, that is the maximum amount of data that can be in the system at some time;
		\item[(b)] the maximum delay, that is the maximum sojourn time of a bit of data. We assume that the system is FIFO per flow: for any given flow, data are served in their arrival order. 
	\end{enumerate}
	
%	(a) the maximum backlog, that is the maximum amount of data that can be in the system at some time; (b) the maximum delay, that is the maximum sojourn time of a bit of data. We assume that the system is FIFO per flow: for any given flow, data are served in their arrival order. 
	
	The {\em backlog  at time $t$}  is $b(t) = A(t)- D(t)$, so  the maximum backlog is $b_{\max} = \sup_{t\geq 0}A(t)- D(t)\leq \sup_{t\geq 0} \alpha(t) - \beta(t)$ : $b_{\max}$ is bounded by the maximum vertical distance between $\alpha$ and $\beta$. 
	
	The {\em delay} of data arriving  at time $t$ is $d(t) = \inf\{d~|~D(t+d) \geq  A(t)\}$, and the maximum delay is $d_{\max} = \sup_{t\geq 0}\inf\{d~|~D(t+d) \geq  A(t)\} \leq \sup_{t\geq 0}\inf\{d~|~\beta(t+d) \geq  \alpha(t)\}$: $d_{\max}$ is bounded by the maximum horizontal distance between $\alpha$ and $\beta$.
	
	This is illustrated in Figure~\ref{fig:perf}. 
	\begin{figure}[htbp]
		\centering
		\subfloat[\label{fig:proc}Processes]{
			\centering
			\begin{tikzpicture}
				\draw[->] (0,0) -- (3,0)node[below, pos = 0.9] {time};
				\draw[->] (0,0) -- (0,2) node[left, pos = 0.8] {\rotatebox{90}{data}};
				\draw[red] (0,0) -- (0, 0.5) -- (1, 1) -- (1.5,1.5) -- (3,1.7) node[above = -0.05cm, pos = 0.8] {$A$};
				\draw[blue] (0,0) -- (0.5,0) -- (0.8333, 0.667) -- (1.5, 1) --  (2,1.5) -- (3,1.63) node[below, pos = 0.5] {$D$};
				\draw[<->, yellow!70!black, thick]   (0.5, 0) node[below]{$t$} -- (0.5, 0.75) node[left, pos = 0.5] {\rotatebox{90}{\tiny $b(t)$}};
				\draw[green!70!black, <->, thick] (1.2,1.2) -- (1.7,1.2) node[above = -0.1cm, pos = 0.75] {\tiny $d(u)$};
				\draw[green!70!black, dotted] (1.7,1.2) -- (1.7, 0) node[below = 0.05cm] {$u$};
		\end{tikzpicture}}
		\hspace{2cm}
		\subfloat[Performances]{
			\centering
			\begin{tikzpicture}
				\draw[->] (0,0) -- (3,0)node[below, pos = 0.9] {$t$};
				\draw[->] (0,0) -- (0,2) node[left, pos = 0.8] {\rotatebox{90}{data}};
				\draw[red] (0,1) -- (3,1.7) node[below = -0.05cm, pos = 0.95] {$\alpha$};
				\draw[blue] (0,0) -- (1,0) -- (3,2) node[below, pos = 0.5] {$\beta$};
				\draw[green!70!black, <->, thick] (0,1) -- (2,1) node[below = -0.1cm, pos = 0.2] {$d_{\max}$};
				\draw[yellow!70!black, <->, thick] (1,0) -- (1, 1.25) node[right = -0.1cm, pos = 0.5] {\rotatebox{90}{$b_{\max}$}};
				%\draw[red!70!black, <->, thick] (0,1.65) -- (2.7, 1.65) node[above = -0.1cm, pos = 0.5] {$\ell_{\max}$};
		\end{tikzpicture}}
		\caption{Processes and worst-case performance.}
		\label{fig:perf}
	\end{figure}
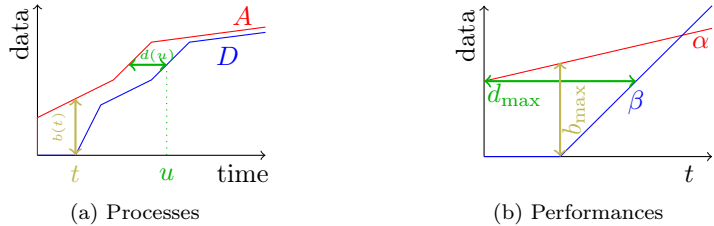
	
	 The last performance index is {\em stability}: that is, deciding if all delays are finite, and backlogs bounded. Roughly, this is not much an issue for tandem network, and it is well-established that the stability condition is that in each server, the arrival rate $r$ of the aggregate process is less than the service rate $R$ (so $r<R$). We assume that this hypothesis holds throughout the paper.

	\begin{example}
		Assume $\alpha = \gamma_{r, b}$ and $\beta = \beta_{R, T}$. If $r\leq R$ , then $b_{\max} = b+rT$ and $d_{\max} = T+\frac{b}{R}$, and $b_{\max} = d_{\max} = +\infty$ otherwise. 
	\end{example}
	
	\paragraph{Network analysis}
	Beyond these first results, the aim of Network Calculus is to compute performance in networks, for example, end-to-end delay upper bounds in complex networks, and maximum backlog in a multi-flow server. The three following results are the basic building blocks for network analysis. When the type of service curve is not specified, it means that any type ((min, plus) or strict) is possible.

	The first block is the propagation of the arrival constraints. In the case of multiple servers in sequences, we need to be able to compute performance not only at the first server, but also at the next server, hence we need to compute the arrival curve at the next server, that is the arrival curve of the departure process. 
	\begin{theorem}
		\label{th:fund}
		Consider a server offering a service curve $\beta$ to a single flow with arrival curve $\alpha$. Then the arrival curve of the departure process is $\alpha' = \alpha\deconv\beta$, where $\forall t\in\R_+$, $\alpha\deconv \beta(t) = \sup_{u\geq 0}\alpha(t+u) - \beta(u)$ is the (min,plus) deconvolution.
	\end{theorem}
	
	Together with  the delay and backlog bounds, Theorem~\ref{th:fund} is often called the {\em fundamental theorem}, and is used in nearly all modular network analysis. 
	
	\begin{example}
		If $\alpha = \gamma_{r, b}$ and $\beta = \beta_{R,T}$, then the departure process is $\alpha\deconv\beta = \gamma_{r, b+rT}$ if $r\leq R$. The arrival rate remains the same as in the arrival curve, but the burst parameter increases. This is an indication of loss of precision in the model. 
	\end{example}
	
	The second block can be seen as a simplification of this step: when a single flow crosses a sequence of servers with respective service curves $\beta_1,\ldots,\beta_n$, it would be more efficient to analyze the end-to-end service directly. We  indeed can compute the {\em end-to-end} service curve.
	\begin{theorem}
		Consider a sequence of $n$ servers offering respective service curves $\beta_1,\ldots,\beta_n$, and a single flow crossing them (one and only once). Then this system guarantees the (min,plus) end-to-end service curve $\beta_{e2e} = \beta_1\conv\cdots\conv \beta_n$ to the flow. 
	\end{theorem}
	\begin{example}
		The concatenation of the two service curves $\beta_{R_1, T_1}$ and $\beta_{R_2, T_2}$ is $\beta_{R_1, T_1}\conv \beta_{R_2, T_2} = \beta_{\min(R_1, R_2), T_1+T_2}$.
	\end{example}
	 
	 The third and last block is when several flows cross a server. In order to compute per-flow performance, we need to compute the {\em per-flow} service curve guarantee for each flow. In the next result we do not assume any service policy that would give some information about the data service order (although assuming FIFO per flow). This is called {\em blind} or {\em arbitrary multiplexing}. Basically, the service curve obtained is the smallest among all possible scheduling.  The following theorem is given for two flows but can be generalized for an arbitrary number of flows by isolating the flow under interest and aggregating all the other flows.
	 
	 \begin{theorem}
	 	\label{th:res}
	 	Consider a server offering a strict service curve $\beta$, that is crossed by two flows 1 and 2, with respective arrival curves $\alpha_1$ and $\alpha_2$. Then the server guarantees the (min,plus) service curve $(\beta-\alpha_2)_+$ to flow 1. 
	 \end{theorem}
	 \begin{example}
	 	If $\beta = \beta_{R, T}$, and $\alpha_2 = \gamma_{r, b}$, then $\beta - \alpha_2 = \beta_{R', T'}$ with $R' = R-r$ and $T' = T+\frac{b+rt}{R-r}$. 
	 \end{example}
	 
	 With these three results, the analysis may become trivial: it case of a feed-forward network (no cycle dependencies) one can analyze the servers in the topological order (from sources to sinks):
	 
	 \begin{enumerate}
	 	\item for each server, where the arrival curves are all known, one can compute 
	 	\begin{enumerate}
	 		\item the residual service curve for each flow,
	 		\item the arrival curves of each departure processes to be able to analyze the next server.
	 	\end{enumerate}
	 	\item Finally,  we obtain a per-flow service curve for each flow, and the (min,plus)-convolution of these service curves per flow is the end-to-end service curve.
	 	\item Performance can be computed. 
	 \end{enumerate}
	  %for each server, for which the arrival curves of all arrival processes are known, one can compute the residual service curve, then the arrival curves of the departure processes to be able to analyze the next server. In the end,  we obtain an individual service curve for each flow, and the (min,plus)-convolution of these service curves per flow is the end-to-end service curve. 
	  But this process, unless used with care,  is usually incorrect, due to the hypotheses that require some type of service curve.

 	\section{The building of a robust Network Calculus theory}
 	The existence of several types of service curves is the main pitfall encountered by the non-expert in Network calculus, but it is also the spice of the theory. 
 	
 	\subsection{The lack of stability of the Network Calculus service curves}
 	\label{sec:ned}
 	\paragraph{Per-flow service curves}
 	From Theorem~\ref{th:res}, it is usually not possible to compute per-flpw service curves with (min,plus) service curves, especially when  blind multiplexing is at work: to see this more precisely,  consider a server offering a (min,plus) service curve $\beta$ ans crossed by two flows $A_1$ and $A_2$ with respective arrival curves $\alpha_1$ and $\alpha_2$. Then, 
 	\begin{align*}
 		(D_1+D_2)(t)&\geq (A_1+A_2)\conv \beta(t)\\
 		& = \min_{0\leq s\leq t}(A_1+A_2)(s)+ \beta(t-s)\\
 		\exists s\leq t,~(D_1+D_2)(t)&\geq (A_1+A_2)(s)+ \beta(t-s)\\
 		\exists s\leq t,~D_1(t)&\geq A_1(s)-[D_2(t)-A_2(s)]+\beta(t-s)\\
 		\exists s\leq t,~D_1(t)&\geq A_1(s)-[A_2(t)-A_2(s)]+\beta(t-s)\\
 		\exists s\leq t,~D_1(t)&\geq A_1(s)-\alpha_2(t-s)+\beta(t-s)\\
 		\exists s\leq t,~D_1(t)&\geq A_1(s)-(\beta-\alpha_2)(t-s)\\
 		D_1(t)&\geq A_1\conv (\beta-\alpha_2)(t). 
 	\end{align*}
 	
 	Since $\beta(0) = 0$, and usually $\alpha_2(0)>0$, the function $(\beta-\alpha_2)$ does not guarantee a non-negative service, hence finite delay. In addition multiple examples exist demonstrating that no service can be guaranteed for service curves with negative values. 
 	
 	\begin{example}
 		\label{ex:neg}
 		If $\alpha_2 = \gamma_{r, b}$ and $\beta = \beta_{R, T}$, then $\beta - \alpha_2$ is not a rate-latency service curve, but is  depicted in Figure~\ref{fig:resneg}. Looking more specifically at the processes, it may happen, as in Figure~\ref{fig:sc}(left) that the server is never empty for flow 2. So if flow 2 has higher priority than flow 1 so there is always always data from flow 2 to serve, and  flow 1 is never served: $D_1=0$. 
 		
 		\begin{figure}[htbp]
 			
 				\centering
 				\begin{tikzpicture}
 					\draw[->] (0,1) -- (3,1)node[below, pos = 0.9] {time} node[pos=0, left]{0};
 					\draw[->] (0,0) -- (0,2) node[left, pos = 0.8] {\rotatebox{90}{data}};
 					\draw[red] (0,1) -- (0, 1.1) -- (3,1.1) node[above = -0.05cm, pos = 0.8] {$A_1$};
 					\draw[blue, dashed] (0,1) -- (0, 0.5) -- (0.5,0) -- (3, 2) node[pos=0.3, left] {$\beta'$};
 					\draw[blue] (0,1) -- (3,1)  node[below, pos = 0.7] {$D_1$};
 					
 			\end{tikzpicture}

 			\caption[bla]{When the residual service curve can have negative values, there is no guarantee for the departure process. }
 			\label{fig:resneg}
 		\end{figure}
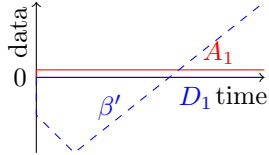
 		
 	\end{example}
 	
 		On the contrary, when the service is strict, we can set $s = start(t)$, use $A_i(s) = D_i(s)\leq D_i(t)$ and compute:
 		\begin{align*}
 			(D_1+D_2)(t)&\geq \max((A_1+A_2)(s), (A_1+A_2)(s)+ \beta(t-s))\\
 			D_2(t)&\geq \max(A_2(s), A_2(s)-[D_1(t)-A_1(s)]+\beta(t-s))\\
 			D_2(t)&\geq \max(A_2(s),A_2(s)-[A_1(t)-A_1(s)]+\beta(t-s))\\
 			D_2(t)&\geq \max(A_2(s),A_2(s)-\alpha(t-s)+\beta(t-s)\\
 			D_2(t)&\geq A_1(s)-(\beta-\alpha)_+(t-s)\\
 			D_2(t)&\geq A_1\conv (\beta-\alpha)_+(t). 
 		\end{align*}
 		The service curves computed is non-negative. However, a (min,plus) service curve is offered to flow 1, and not a strict service curve, that could be requested for a subsequent analysis. 
 		
 		Nevertheless, yet another difficulty arises, with the concatenation. 
 		
 	\paragraph{Lack Robustness of strict service curves regarding concatenation. }
 	Stability issues regarding per-flow service curves is a long-standing problem~\cite{BJT10b}, already studied. We can notice that the service curve we obtain is a (min,plus) service curve, for which the same type of computations cannot be done anymore. 
 	A similar problem arises when dealing with concatenation: the concatenation of two servers with respective service curves $\beta_1$ and $\beta_2$ is a system offering the (min, plus) service curve $\beta_1\conv\beta_2$. But the concatenation of two strict service curves results only in a (min,plus) service curve. Results for computing a residual service curve are then not valid anymore~\cite{Bou11}. 
 	
 	\subsection{Solutions to circumvent these issues}
 	
	\paragraph{More global analyses}
	Fortunately, efficient solutions have been found to perform correct analyses of networks. The first solution is to reverse the order in which the analysis is done: the classical analysis we have just discussed is aiming at finding arrival curves first and  manipulating them as much as possible and only them (cumulative processes are never manipulated as a mathematical object). 
	
	On  the contrary, more advanced methods use the cumulative processes in the analysis, and do the conversion to service and arrival curve in the last step. This is what is  done in the PMOO analysis ({\em Pay Multiplexing only once}). This phenomenon has been exhibited first in the case of two servers in concatenation crossed by two flows, by noticing that a bit of data of the low-priority flow can be overtaken only once by a bit of data of the higher priority flow. Hence it was be possible to improve the analysis~\cite{SZM08} by factorizing the computations into a {\em multidimensional (min,plus)-convolution}~\cite{BGLT2007}. However, this solution is also not perfect, since limited to tree networks~\cite{BNS22b} and not optimal~\cite{SZF08}. Using linear programming, it is possible to compute nearly exact worst-case delay bounds with blind multiplexing in general feed-forward networks~\cite{BJT10} and an adaptation (although less accurate) is possible for networks with cyclic dependencies~\cite{Bou22}.
	
	\paragraph{Using specific service policies}
	Blind multiplexing has been widely used at first, but with the development of critical networks, it has been proved is very pessimistic: indeed, it only computes the performances for the worst-case scheduling of the flow of interest. Critical networks are usually tightly controlled, networks are multiclass with precise description of the scheduling policy. This offers the opportunity of precise modeling, hence more accurate performance bounds. It is usually done as follows: using the strict service curve, it is possible to compute per-class (min,plus) service curves. Inside each class of traffic, the service order is FIFO, the scheduling policy where strict service curves are no required. 
	
	%The first scheduling that have been modeled are the FIFO and static priorities, where the analysis can be found in the two seminal books~\cite{Chang2000, LT2001}. Other scheduling policies have been studied for a long time, like EDF ({\em Earliest-deadline-First})~\cite{Chang2000} and GPS ({\em Generalized processor sharing}). This latter policy is the idealization of the bandwidth sharing policies, and has more recently received some attention. Practical implementations of GPS, like DRR {\em Deficit Round Robin} and WRR ({\em Weighted Round Robin}) have been extensively studied in the recent years, leading to huge improvements of the performance computation in realistic networks. 
	
	Usually, the network element involved are or two main types: either they offer a constant rate service, that can be modeled by a strict service curve, or they guarantee a maximum transmission delay modeled. This latter case is intermediate between (min,plus) and strict service curve, to be discussed in Section~\ref{sec:pmoo}. 
	In both case the modeling is not an issue for Network Calculus. 
 	\paragraph{Solution base on accepting service curves with negative values}

	More recently, a solution has been proposed~\cite{HCS24} to perform the analysis without  considering  strict service curves. The key idea is allowing service curves with negative values (and local decrease).  If (min,plus) per-flow service curve can be obtained from (min,plus) service curves, then concatenation also holds, and all results can be applied straightforwardly in any order. Although multiple examples, such as Example~\ref{ex:neg}, show that no service can be guaranteed under the described model, the authors add an additional constraint: the flow under analysis must have a minimum arrival curve $\alpha_{\ell}$: $\forall 0\leq s\leq t$, $A(t)-A(s) \geq \alpha_{\ell} (t-s)$, and $\alpha_{\ell}$ is required to grow to $+\infty$ with $t$  (which was what prevented our example to guarantee any traffic). 
	
	%The key principle of this idea is to make the network analysis doable, without having to worry too much about the hypothesis: if (min,plus) residual service curve can be obtained from (min,plus) service curves, then concatenation also holds, and all results could be applied without stability issue. 
	
	 In this direction,~\cite{HCS24} computes the (min,plus) service curves $\beta-\alpha$ as the per-flow  service curves  by allowing negative values and handles the computation of  worst-case delay and backlog upper bounds that depends on the minimum arrival curve $\alpha_{\ell}$.
	 %The key idea is assuming a minimal arrival guarantee of the arrival processes: if data arrive at rate at  least $r$, then data cannot be stored infinitely often in the system, and finite delays can be computed.  
	 
	 \begin{example}
	 		Assume a rate-latency service curve $\beta(t) = R(t-T)_+$ crossed by flow 1 and flow 2, respectively constrined by $\alpha_1$ and $\alpha_2$ with $\alpha_i = \alpha_{r_i, b_i}$, $i\in\N_2$. Assume in addition that $A_2(t)-A_2(s)\geq r'_2(t-\tau_2)_+$ for all $0\leq s\leq t$. Then a residual service curve for flow 2  is $\beta-\alpha_1$, and the worst-case delay is
	 		\begin{equation}
	 			\label{eq:delay-nc}
	 			\max(T+\frac{b_1+r_1T+b_2}{R-r_1}, T+\tau_2 +\frac{b_1+r_1T}{r'_2}).
	 		\end{equation}
	 		 
	 		  We can notice that when $r'_2$ decreases to 0, the worst-case delay grows to $+\infty$.
	 	\end{example}   
	 		This type of service curve, under the additional constraint of a minimum arrival curve allows to compute performances in networks without any assumption on the type of service curve. However, the analysis is restricted to the arbitrary multiplexing only:  for a precise modeling and accurate performance bounds,  many scheduling policies require strict service curve as an assumption, that is used throughout the analysis. 
	 
%	 For example, assume a rate-latency service curve $\beta(t) = R(t-T)_+$ crossed by flow 1 and flow 2. Flow 1 is $\alpha_1$-constrained with $\alpha_1 = \alpha_{r_1, b_1}$ and flow 2 is $\alpha_2$-constrained with $\alpha_2 = \alpha_{r_2, b_2}$. Assume in addition that $A_2(t)-A_2(s)\geq r'_2(t-\tau_2)_+$ for all $0\leq s\leq t$. Then a residual service curve for flow 2  is $\beta-\alpha_1$, and the worst-case delay is $$\max(T+\frac{b_1+r_1T+b_2}{R-r_1}, T+\tau_2 +\frac{b_1+r_1T}{r'_2}).$$ We can notice that when $r'_2$ decreases to 0, the worst-case delay grows to $\infty$. 
%	This type of service curve, under the additional constraint of a minimum arrival allows to compute performances in networks without any assumption on the type of service curve. However, the analysis is restricted to the arbitrary multiplexing only, since many scheduling policies require strict service curve as an assumption, that is used throughout the analysis. 
	In~\cite{HCS24}, two use-cases presented: feedback control and analysis of tandem networks. Complex feedback structures are another use-case described in~\cite{HWS25}. 
	
	We claim that these use-cases could be analyzed with a more conventional analysis, by studying a class of sub-additive functions, that received to our knowledge very little attention regarding service curves. That is the purpose of the  next section. 
	
	%Before analyzing these case, we present a special shape of service curve that has received little attention until now but that could be useful for the study of these systems: the sub-additive service curves. In the next section, we show that in fact,  it is possible to compute individual service curves for  sub-additive (min,plus) service curves. 
	
	\section{Sub-additive service curves}
	\label{sec:sub-add}
	In this paragraph, we consider a sub-additive (min,plus) service curve, and show that non-negative per-flow service curves can be guaranteed to each flow under the blind multiplexing scheduling. 
	Let us first give some definitions. 
	
	\begin{defi}[sub-additive function]
		Let $f:\R_+\to\Rmin$ be a function. This function is sub-additive if for all $t, s\in\R+$, $f(s+t)\leq f(s)+f(t)$. 
	\end{defi}
	
	Examples of sub-additive functions are: 
	\begin{itemize}
		\item $\gamma_{r, b}$, the token-bucket functions;
		\item $\lambda_R$, the constant-rate functions.
	\end{itemize}
	
Sub-additive functions have received some attention regarding arrival curves, since if $\alpha$ is an arrival curve for the process $A$, then its sub-additive closure (defined next) also is. On the contrary, it is usual to assume that a {\em strict} service curve is {\em super-additive} ($\forall s, t$, $\beta(s+t)\geq \beta(s)+\beta(t)$), since its can be replaced by its super-additive closure.

	\begin{defi}
		The sub-additive closure of $f$ is $f^* = \min_{n\in\N}f^{n}$, where $f^0 = \delta_0$ and for all $n>0$, $f^{n} = f\conv f^{n-1}$. The sub-additive function $f$ is sub-additive and is the largest sub-additive function $g$ such that $g\leq f$. 
	\end{defi}
	
	The following lemma will be useful.
	
	\begin{lemma}[Prop. 2.9 of \cite{BBL18}]
		$f=f^*\Leftrightarrow f \deconv f = f$.
	\end{lemma}
	
	In this section, we assume a server offering a (min,plus) non-negative, non-decreasing and left-continuous service curve $\beta$ that is sub-additive. We then know than $\beta(0) = 0$, and can assume that $\beta = \beta\deconv\beta $. 
%	\begin{proof}
%		Since $f$ is sub-additive, for all $0\leq s\leq t$, $f(s)+f(t) \geq f(s+t)$, and $f(0) = 0$ (since $f$ is non-negative).  Therefore, $f(t+s)-f(s) \geq f(t)$, and $f(t)\leq \sup_{s>0}f(t+s)-f(s)$, taking into account the case $s = 0$, $f(t)= \sup_{s\geq 0}f(t+s)-f(s)$. 
%	\end{proof}

	%We can then assume that $\beta = \beta\deconv\beta $. 

	%In this paragraph, we show that it is possible to compute a non-negative  residual service curve. 
	%To do this, we first show that we can bound the length of the backlogged period, and that $\beta$ is  in fact a strict service curve in the weak sense.  
	Consider two arrival processes to the server, $A_1$ and $A_2$, with respective departure processes $D_1$ and $D_2$ and respective arrival curves $\alpha_1$ and $\alpha_2$. Our goal is to compute a service curve  for flow 2.  
	
	% We have: 
%	
%	
%	\begin{align*}
%		D_1(t)+D_2(t) & \geq (A_1+A_2) * \beta(t)\\
%		& = \inf_{0\leq s\leq t}A_1(s)+A_2(s) + \beta(t-s)\\
%		D_2(t) &\geq  \inf_{0\leq s\leq t} A_1(s)+A_2(s) + \beta(t-s) - A_1(t), 
%	\end{align*}
	
	For the sake of simplicity, let us denote $A = A_1 + A_2$ and $D = D_1+D_2$, we than have for all $t\geq 0$, $D(t) \geq A*\beta(t) = \inf_{0\leq s\leq t}A(s)+\beta(t-s)$.
	
	As $\beta$ and $A$ are left-continuous, there exists $s\in[0, t]$ such that $A*\beta(t) = A(s)+
	\beta(t-s)$ (\cite[Prop 3.10]{BBL18}).
	We define $u(t)$ as $u(t) = \sup\{s\leq t~|~A*\beta(t) = A(s)+\beta(t-s)\}$, the supremum of the values where the inequality holds.
	
	\begin{lemma}
		\label{lem:u}
		$\forall t\geq 0$, $A\conv \beta(t) = A(u(t))+ \beta(t-u(t))$. 
	\end{lemma}

	\begin{proof}
		Let $\cU = \{u\in[0, t]~|~A\conv\beta(t) = A(u) + \beta(t-u)$. If $u(t)\in\cU$, then there is nothing to prove. If $u\notin\cU$, since $u(t) = \sup \cU$ then there exists an increasing sequence $(u_n)_{n\in\N}\in\cU^{\N}$ such that $\lim_{n\to\infty}u_n = u(t)$: for all $n\in\N$, 
		$A(u_n) + \beta(t-u_n) = A\conv\beta(t) < A(u(t)) + \beta(t-u(t))$, or by rewriting, 
		
		\begin{equation}
			\label{eq:lem2}
			A(u_n) - A(u(t)) < \beta(t-u(t)) - \beta(t-u_n).
		\end{equation}
		By left-continuity of $A$, $\lim_{n\to\infty} A(u_n)-A(u(t)) = 0$, and because $\beta$ is non-decreasing,\\ $\lim_{n\to\infty} \beta(t-u(t)) - \beta(t-u_n)\leq 0$. But if Equation~\eqref{eq:lem2} holds, we necessarily have $\lim_{n\to\infty} \beta(t-u(t))- \beta(t-u_n)= 0$, hence $u(t)\in\cU$. This is a contradiction and  ends the proof. 
	\end{proof}
	
	We are now interested in the minimal departure process, that is  $D= A\conv \beta$. We will later see that it is not a loss of generality. 
	\begin{lemma}
		\label{lem:limit}
		If $D = A\conv \beta$, then for all, $s\leq t$, $D(t)-D(s)\leq \beta(t-s)$.
	\end{lemma}
	
	%We are now interested in the minimal departure process, that is that we from now on assume that $D= A\conv \beta$. 
	
	\begin{proof}
		From Lemma~\ref{lem:u}, $D(s) = A(u(s)+\beta(s-u(s))$, and $D(t)\leq A(u(s))+\beta(t-u(s))$ always holds. Consequently $D(t)-D(s) \leq \beta(t-u(s)) - \beta(s-u(s)) \leq \beta(t-s)$ ($\beta$ is sub-additive, so $\beta(t-u(s)) \leq \beta(t-s) + \beta(s-u(s))$). 
	\end{proof}
	In other words, the speed of the departure process is limited. The next property is interesting as it proves that $t-u(t)$ is upper bounded,  allowing to focus only on a finite bounded interval for the possible values of $u(t)$.  
	%\Define $u(t)$ as $u(t) = \sup\{s\leq t~|~A*\lambda_R(t) = A(s)+R(t- s)\}.$
	%\begin{lemma}
	%	\label{lem:u-inc}
	%	$t\mapsto u(t)$ is non-decreasing.
	%\end{lemma}
	%\begin{proof}
	%	On the one hand, we have 
	%	$$\begin{array}{ll}
		%		&D(t) = A(u(t))+\beta(t-u(t))\\
		%		-&D(s) = A(u(s)+\beta(s-u(s))\\
		%		\hline
		%		&D(t)-D(s) = A(u(t))-A(u(s))+ \beta(t-u(t) - \beta(s-u(s)) .
		%	\end{array}$$
	%	Let us now assume, by contradiction,  that $t>s$ and $u(t)<u(s)$. 
	%	On the other hand, since $t>s$, then $D(t)- D(s)\leq \beta(t-s)$, and since $u(t)<u(s)$, $A(u(t))-A(u(s)\leq 0$, and since  $\beta$ is non-decreasing and sub-additive, $\beta(t-u(t)-\beta(s-u(s))\leq \beta(t-s-(u(t)-u(s))), 
	%	
	%	 so $D(t)-D(s) = A(u(t))-A(u(s))+ \beta(t-u(t) - \beta(s-u(s)) \leq A(u(t))-A(u(s)) + \beta(t-s-u(t)+u(s))$. If $u(t)<u(s)$, then $A(u(t))-A(u(s))\leq 0$ and $t-s -u(t)+ u(s)>t-s$, so $D(t)-D(s) \geq \beta(t-s)$. The only possibility is than $A(u(t)) = A(u(s))$ and which is in contradiction with the previous lemma. If $A(u(t) = A(u(s))$ and $\beta(t-s) = \beta(t-s -u(t)+ u(s))$ both hold, then $A(u(s)) + \beta (t-u(s)) \leq A(u(s)) + \beta(t-s)?? = D(t)$
	%\end{proof}
	
	Let us define $\tau = \sup\{t>0~|~\alpha(t)\geq\beta(t)\}$, the time from which $\alpha$ is always strictly below $\beta$. 
	
	\begin{lemma}
		\label{lem:tau}
		For all $t\geq 0$, $t-u(t)\leq \tau$. 
	\end{lemma}
	\begin{proof}
		Assume on the contrary that for some $t\geq 0$, $t-u(t)>\tau$. Since $u(t)\neq t$, $A(t)>D(t)$.  Then $A(t) > A(u(t))+\beta(t-u(t))$, and $\alpha(t-u(t)) \geq A(t)-A(u(t))>\beta(t-u(t)$. But, by construction, for all $s>\tau$, $\alpha(s) <\beta(s)$. We have a contradiction, and $t-u(t)\leq \tau$.
	\end{proof}
	We next show that the first backlogged period of  the server is  bounded: there is some time $t$ where the server is empty. Note that this result is now valid even if $D\neq A\conv  \beta$: $\forall t\geq 0$, $A(t)\geq D(t)\geq A(t) = A\conv\beta(t)$, so if $A(t) = A\conv\beta(t)$, then $A(t) = D(t)$.
	Let us first prove the following lemma:
	
	\begin{lemma}
		\label{lem:start}
		For all $t\geq 0$, $D(u(t)) = A(u(t))$. 
	\end{lemma}
	\begin{proof}
		We proceed by contradiction and assume that there exists $t$ such that $D(u(t))<A(u(t))$. By construction, $u(t)>0$, and $u(u(t)) < u(t)$. We have
		
		\begin{align*}
			D(t) &= A(u(t))+\beta(t-u(t))\\
			&> D(u(t))+\beta(t-u(t))\\
			& = A(u(u(t))) + \beta(u(t)-u(u(t))+\beta(t-u(t))\\
			& \geq A(u(u(t)) + \beta(t-u(u(t))),
		\end{align*}
		where in the last line we use the sub-additivity of $\beta$. Finally, $D(t) > A(u(u(t)) + \beta(t-u(u(t)))$, showing that $u(t)$ is not a minimizer. This is a contradiction and 
		$D(u(t)) = A(u(t))$. 
	\end{proof}
	
	We are now ready to prove that the first backlogged period is bounded. 
	
	\begin{lemma}
		\label{lem:stb}
		There exists $t>0$ such that $D(t) = A(t)$. 
	\end{lemma}
	
	\begin{proof}
		From Lemma~\ref{lem:tau}, $t-u(t)\leq \tau$ for all $t\geq 0$. Fix $t>\tau$. Then $u(t)>0$, and from Lemma~\ref{lem:start}, $A(u(t)) = D(u(t))$, which concludes the proof. 		
%		We proceed again  by contradiction and assume that for all $t>0$,
%		$A(t)> D(t)$, hence $u(t) <t$ for all $t>0$. Fix $t>0$. Bin, since $t-u(t)$ is bounded, one can assume that $t$ is large enough  so that $u(t)>0$, and .  By assumption,  $u(t)<t$, $D(u(t))<A(u(t))$ and $u(u(t))<u(t)$.  Then, we have
%		\begin{align*}
%			D(t) &= A(u(t))+\beta(t-u(t))\\
%			&> D(u(t))+\beta(t-u(t))\\
%			& = A(u(u(t))) + \beta(u(t)-u(u(t))+\beta(t-u(t))\\
%			& \geq A(u(u(t)) + \beta(t-u(u(t))),
%		\end{align*}
%		where in the last line we use the sub-additivity of $\beta$, and we obtain that $u(t)$ is not a minimizer for the convolution. This is a contradiction and if $A(s)>D(s)$ for all $s\in]0, t]$ then $u(s) = 0$, and, from the previous lemma, this can hold only on an interval of length at most $\tau$, and there exists $t\in[0, \tau] $ such that $A(t) = D(t)$. 
	\end{proof}
	
%	This lemma shows that the first backlogged period is bounded. It is now  not difficult to show that all backlogged periods are bounded and to see that a individual service guarantees can be computed similar to the strict service curve case. 
%	
%	\begin{lemma}
%		\label{lem:start}
%		For all $t\geq 0$, $D(u(t)) = A(u(t))$. 
%	\end{lemma}
%	
%	\begin{proof}
%		The proof follows the same line as the proof of the previous lemma: if $D(u(t))<A(u(t))$, then  $u(u(t))< u(t)$ and $D(t) > A(u(u(t)) + \beta(t-u(u(t)))$, which is a contradiction and completes the proof.  
%	\end{proof}
	
	And we are now ready to prove the main result of this section.
	
	\begin{theorem}
		\label{th:sb-pf}
		Consider a server offering a non-decreasing, non-negative, left-continuous and sub-additive (min,plus) service curve, crossed by two flows  constrained by the respective arrival curves $\alpha_1$ and $\alpha_2$. Then $(\beta-\alpha_2)_+$ is a (min,plus)  service curve for flow 1. 
	\end{theorem}
	
	\begin{proof}
		Let $A_1$ and $A_2$ the respective cumulative arrival processes for flows 1 and 2, and $D_1$ and $D_2$ their respective cumulative departure processes. We have for all $t\geq 0$:
		
		$$D_1(t)+D_2(t)  \geq A_1(u(t))+A_2(u(t)) + \beta (t-u(t)). $$
		On the one hand, 
	%	Consider 2 flows crossing the system:
		\begin{align*}
			D_1(t)& \geq A_1(u(t))+A_2(u(t)) + \beta (t-u(t))- D_2(t)\\
			& \geq  A_1(u(t) + (\beta(t-u(t)) - A_2(t) + A_2(u(t))) \\
			& = A_1(u(t) + (\beta(t-u(t)) - \alpha_2(t-u(t)),
		\end{align*}
		and on the other hand, $D_1(t) \geq D_1(u(t)) = A_1(u(t))$, so
		$$D_1(t)\geq A_1(u(t) + [(\beta(t-u(t)) - \alpha_2(t-u(t))]_+. $$
		In conclusion,  $D_1\geq A_1\conv (\beta-\alpha_2)_+$. 
	\end{proof}
	
	\begin{example}
		\begin{enumerate}
			\item The function $\lambda_R:t\to R t$, the constant rate is sub-additive, and it is possible for this classical class of service curves to compute residual service curves.
			\item The function $f:0\to 0;~t\to W+R(t-T)_+$ is sub-additive if $f\deconv f = f$, that is (the only interesting case is $s\leq T$ and $t> T$), $f(s)+f(t) = W + W+R(t-T) \geq f(t+s) = W+R(t+s-T)\Leftrightarrow W\geq Rs$, so if $W\geq RT$. This function is useful for window-flow control. Figure~\ref{fig:ex-sub} shows an example of process and arrival curves. We see that when the system is empty ($A=D$), then $u(t) = t$, but when $A(t)>D(t)$, then $u(t) = 0$ (we only consider the first backlogged period). 
		\end{enumerate}
		
		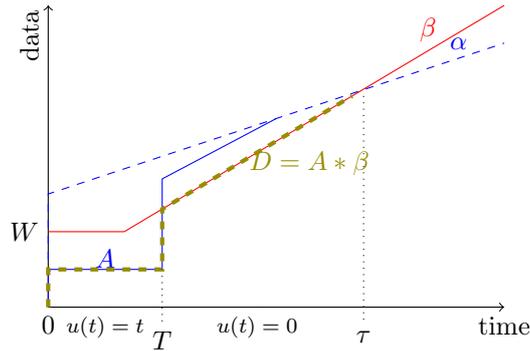
\begin{figure}[htbp]
			\centering\begin{tikzpicture}
				\draw[->] (0, 0) -- (6, 0) node[below, pos = 0] {0} node[below, pos = 1] {time};
				\draw[->] (0, 0)-- (0, 4) node[left, pos=0.9] {\rotatebox{90}{data}};
				\draw[red]  (0, 1) -- (1, 1) node [black, pos=0, left] {$W$}-- (6, 4)node[pos=0.8, above] {$\beta$};
				%\draw[red] (0, 0) -- (6, 4);
				\draw[blue, dashed] (0, 0) -- (0, 1.5) -- (6, 3.5)node[pos=0.9, above] {$\alpha$};
				\draw[blue] (0, 0) -- (0, 0.5)  -- (1.5, 0.5) node[pos=0.5, above=-.1cm] {$A$}-- (1.5, 1.7) -- (3,2.5) ;
				\draw [dotted] (1.5, 0.5) -- (1.5, -0.2)node [below, pos=1] {$T$};
				\draw [dotted] (0.75, 0) -- (0.75, 0) node [below, pos=0.5] { {\footnotesize $u(t) = t$}} ;
				\draw [dotted] (4.15, 2.8) -- (4.15, -0.2) node [below, pos=1] {$\tau$};
				\draw [dotted] (2.75, 0) -- (2.75, 0) node [below, pos=0.5] { {\footnotesize $u(t) = 0$}} ;
				
				\draw[olive, dashed, ultra thick] (0, 0) -- (0, 0.5)  -- (1.5, 0.5) -- (1.5, 1.3)-- (4, 2.8) node[pos = 0.4, right] {$D  = A\conv \beta$}; 
			\end{tikzpicture}
			\caption {Example of a departure process with a sub-additive service curve. First, $\beta >A$, and necessarily, $A = D$. Then, $A<\beta$, and necessarily, $u(t) = 0$ and $D = \beta$.}
			\label{fig:ex-sub}
		\end{figure}
	
	\end{example}
	
	\paragraph{Comparison with the weakly strict service curves.} This result could suggest to the fast reader that a sub-additive (min,plus) service curve is also a weakly strict service curve ($\forall t\geq 0$, $D(t)\geq D(start(t)) + \beta(t-start(t))$), and then be in contradiction with~\cite{BJT10} stating that the only (min,plus) service curves that are also weakly strict service curves are only the trivial functions ($\delta_0$ and the zero function). 
	Here we show that this is not true: sub-additive (min, plus)  service curves are not weakly strict service curves (with the exception ot the two trivial functions). We recall and illustrate the proof in the next example. 
	
	\begin{example}
		First, if $\beta$ is sub-additive and non-zero, then for all $t>0$, $\beta(t)>0$
		Consider a {\em large enough} arrival curve $\alpha$, such that $\alpha(t)>\beta(t)$ for all $t\in (0, \tau]$, for some  $\tau >0$. 
		Let us construct  the arrival process $A$ with $A(t) = \alpha(t)$, for all $t\in[0, \tau]$, and build two the departure processes $D_{ws}$, when the service is weakly strict and $D_{mp}$ when the service curve is (min, plus). For all $t\in[0, \tau/2[$, define $D_{ws}(t) = D_{mp}(t) = \beta(t)\geq A\conv\beta(t)$, and set $D_{ws}(\tau/2) = D_{mp}(\tau/2)= A(\tau/2)$. 
		
		Then $(0, \tau/2)$ is a backlogged period, and for $t\in(\tau/2, \tau)$: $$D_{ws}(t)\geq A(\tau/2) + \beta(t-\tau/2) = \alpha(\tau/2) + \beta(t-\tau/2)>\alpha(\tau/2).$$
		during the second backlogged period. 
		
		% according to the lemmas above, since there is only one backlogged period in $[0, \tau/2]$, and set $D(\tau/2) = A(\tau/2)$. This a new start of backlogged period. 
		
		%If the service curve were strict, one would have for all $t\in [\tau/2, \tau]$, $$D(t)\geq A(\tau/2) + \beta(t-\tau/2).$$ 
		
		Next, if the service curve is (min,plus), for $t\in(\tau/2, \tau)$, 
		$$D_{mp}(t) = \min(A(\tau/2) + \beta(t-\tau/2), \beta(t)) \lor A(\tau/2) = \min(\alpha(\tau/2)+ \beta(t-\tau/2), \beta(t))\lor \alpha(\tau/2).$$
		%and $D_{mp}$ remains constant 
		
		The two notions can coincide only if $D_{mp} \geq D_{ws}$. But, as $\beta$ is sub-additive, $\beta(t)\leq \beta(\tau/2)+\beta(t-\tau/2)<\alpha (\tau/2)+\beta(t-\tau/2)$, $D_{mp}$ remains constant while $\beta(t) <\alpha(\tau/2)$, and during that time, $$D_{mp}(t)  = \beta(t)\leq \alpha(\tau/2)<D_{ws}(t).$$
		
		%The two notions can coincide only if $D_{mp} \geq D_{ws}$. 
		 
		 %But as $\beta$ is sub-additive, $\beta(t)\leq \beta(\tau/2)+\beta(t-\tau/2)<\alpha (\tau/2)+\beta(t-\tau/2)$. 
		
		%is a possible departure process. But then we cannot have and the two notions of service curve coincide if $\beta(t)\geq  A(\tau/2)+\beta(t-\tau/2)$, but we have $A(\tau/2) + \beta(t-\tau/2) > \beta(t/2) + \beta(t-\tau/2) \geq \beta(t)$, where we use that $A(\tau/2)>\beta(\tau/2)$ and the sub-additivity of $\beta$. As a consequence, the two notions of service curve do not coincide. 
		
		This illustrated in Figure~\ref{fig:weak} for a constant rate service curve. 
		
		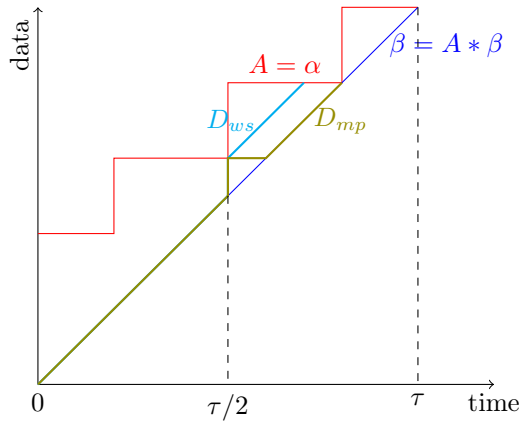
\begin{figure}[htbp]
			\centering
			\begin{tikzpicture}
				\draw[->] (0, 0) -- (6, 0) node[below, pos = 0] {0} node[below, pos = 1] {time};
				\draw[->] (0, 0)-- (0, 5) node[left, pos=0.9] {\rotatebox{90}{data}};
				\draw[red] (0,2) -- (1, 2) -- (1, 3)-- (2.5, 3) -- (2.5, 4) -- node[pos=0.5, above] {$A = \alpha$} (4, 4)-- (4, 5)-- (5, 5);
				\draw[blue] (0, 0) -- (5, 5) node [pos=0.9, right]{$\beta = A\conv \beta$} ;
				\draw[cyan, thick] (0, 0) -- (2.5, 2.5)-- (2.5, 3) -- (3.5, 4) node[left, pos=0.5] {$D_{ws}$};
				\draw[olive, thick] (0, 0) -- (2.5, 2.5)-- (2.5, 3) -- (3, 3)-- (4, 4) node[right, pos=0.5] {$D_{mp}$};
				\draw[ dashed] (5, 5) -- (5, 0) node[below, pos=1] {$\tau$};
				\draw[ dashed] (2.5, 2.5) -- (2.5, 0) node[below, pos=1] {$\tau/2$};
			\end{tikzpicture}
			\caption{Sub-additive  (min, plus) service curves are not weakly strict service curve: With $\beta=\Lambda_R$. The departure process $D_{ws}$, for the weakly strict service curve, forces a service at rate $R$ and is strictly above the convolution when creating  the end of a backlogged period: for the (min,plus) service,   $D_{mp}$ is only guaranteed to be above $A\conv \beta$.}
			\label{fig:weak}
		\end{figure}
	\end{example}
	
	As a last remark, It might be useful to notice that $(\beta-\alpha_2)_+$ is not a sub-additive function.

	\section{Use-case 1: Computation-communication systems}
	\label{sec:pmoo}
	We now focus on the first use-case of~\cite{HCS24}: the analysis of a tandem network. 
	 The network is a sequence of $n$ computation/communication (C/C) components. Each C/C component $j$ is composed of a computational element acting as a variable delay $d\in[m_j, M_j]$ and a communication element that is a (min,plus) constant-rate service curve $\beta_j = \lambda_{R_j}$. The authors claim that {\em conventional analysis} is not possible. On the contrary, we demonstrate that, since $\beta_j$ is linear, hence sub-additive, it is possible. to make such analysis. But first, let us comment about a caveat of the pure delay service curve, that is conform to the interpretation of the authors of~\cite{HCS24}. 
	
	\subsection{Interpretation of pure delay curves}
	In the common interpretation, a server with a pure delay service curve $\delta_d$ delays each bit of data by a delay at most $d$. It is easy to see that $\delta_d$ is not sub-additive as soon as $d>0$. Nevertheless, with this interpretation,  since de delay bound $d$ is valid for all bits of data, it is valid for any flow, and we have per-flow service curve $\delta_d$ is valid for all flow $i$ crossing that server:
	
	for all $t\geq 0$, $$D_i(t) \geq A_i([t-d]_+),$$ 
	add working with functions with negative values is useless. 
	
	We now show that this common interpretation is both too strong to be a (m,in,plus) service curve, an too  weak to be a  strict service curve. 
	% This is a case where it could be useful to work with service curve with negative values (we can also note that this service curve is not sub-additive). Let us comment more precisely the interpretation as a (min,plus) service curve and a strict service curve.
	
	\paragraph{Transmission delay is not a (min,plus) service curve}
	Consider a server offering a (min,plus) service curve $\delta_d$, $d>0$, crossed by two flows, with respective arrival curves $\alpha_1$ and $\alpha_2$, and arrival and departure processes $A_1$, $A_2$, $D_1$, and $D_2$. We have the aggregate processes $A = A_1+A_2$ and $D = D_1+D_2$. The (min,plus) service curve imposes $D\geq A\conv \delta_d  = A((\cdot-d)_+)$: graphically, the minimal departure process is shifted on the right by $d$. 
	
	Consider the following example:
	\begin{example}
		Flow 1 has a higher priority  than flow 2 and $A_1 = \alpha_1 = \gamma_{r_1, b_1}$. Assume that $A_2(t)<rd$ for all $t\leq 0$. In total, for all $t\geq d$, $D(t) = A_1(t-d) + A_2(t-d) = b_1+ r_1(t-d) + A_2(t-d) < A_1(t)$: the aggregate departure process is always strictly less than the arrival process of flow 1, hence flow 1 is always backlogged, as shown in Figure~\ref{fig:delay}. Consequently, flow 2, having the least priority, never receives any service, and then has a non-bounded delay: the server is not a transmission delay. 
	\end{example}
	One can argue that this is precisely a case where service curves with negative values could be useful. Yet, realistic server offering such a service curve is still, to out knowledge, to be modeled. 
	
	\begin{figure}[htbp]
		\centering
		\begin{tikzpicture}
			\draw[->] (0, 0) -- (6, 0) node[below, pos = 0] {0} node[below, pos = 1] {time};
			\draw[->] (0, 0)-- (0, 5) node[left, pos=0.9] {\rotatebox{90}{data}};
			\draw[red] (0,2) -- (5.5, 4.5)  node[pos=0.1, below] {$A_1 $} ;
			\draw[red, dashed] (0,2) -- (5, 4.5)  node[pos=0.8, left] {$A_1+A_2 $} ;
			\draw[blue, thick] (0, 0) -- (1, 0)  -- (1,2)node[pos=0.8, right] {$D_1 = D_1+D_2\geq A_1+A_2 \conv \delta_2$} -- (6, 4.5) ;
			
		\end{tikzpicture}
		\caption{Pure delay  (min, plus) do now guarantee the service of flow 2, the lower priority flow.}
		\label{fig:delay}
	\end{figure}
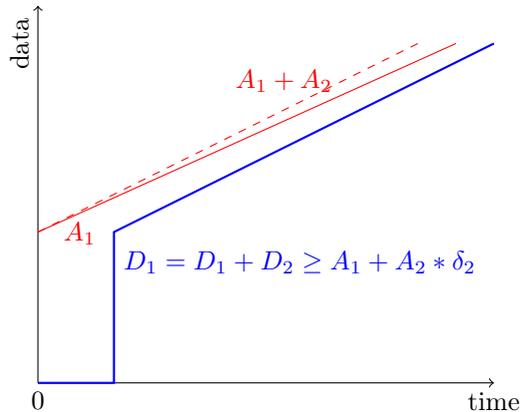
	
\paragraph{Transmission delay is not a strict service curve}
If the server now offers a strict service curve, with the same notations as the paragraph above, we have $$D(t)\geq A(start(t)) +\delta_d(t-start(t).$$ Since $\delta_d(d) = 0$, there is no guarantee for the service in the time interval $[start(t), start(t) +d]$, but for all $t>d$, $\delta_d(t) = +\infty$, so there cannot be backlogged period of length strictly more than $d$, whatever the flow considered. The guarantee is then the same for both flows: they have a succession of backlogged periods of length at most $d$. This forces bits of data arriving just before $d$ to have a near-zero delay, which is not the model for transmission delay (especially if the transmission delay is lower-bounded). 

	\medskip
	
	Back to the model of transmission delay, the authors of~\cite{HCS24} use a minimum delay $m$ and a maximum delay $M$: $A\conv \delta_M \geq D(t) \geq A\conv \delta_m$. 
	
	 A minimum delayredices the indeterminism in the processes, and a way for a more accurate analysis is to consider service curves with delays $M-m$ and add a delay $m$ (the deterministic, minimum guaranteed delay) in the final step of the end-to-end delay computation. We can also use this trick, as it holds in our analysis.

	\subsection{The tandem network model}
	In this paragraph, we show that the conventional PMOO analysis holds when the tandem network is a mix of transmission delays, (min,plus) sub-additive and strict service curves. The analysis is not new, since it mimicks those presented in~\cite{BGLT2007} and~\cite{BNS22b} (the analysis could be extended to sink trees as well, but we do not present it here; the adaptation would be straightforward). 
	
	%We will show in this paragraph that the analysis from \cite{Jens} can be performed using {\em conventional} (slightly updated) analysis. in fact, we show a
	
	Let us now remind our tandem network model with $m$ flows and $n$ servers.  
	
	\begin{itemize}
		\item each flow $i$ is constrained by the arrival curve $\alpha_i$;
		\item the path of flow $i\in \N_m$ is $\pi_i = \langle f_i, \ldots, \ell_i\rangle$; the path of flow 1, the flow of interest, is $\pi_1 = \langle 1, \ldots, n \rangle$. 
		\item each server $j$ offers a service curve $\beta_j$ following one of the three models:  
			\begin{enumerate}
				\item[a)] $\beta_j$ is a pure delay service curve with the interpretation of a transmission delay, where the delay of each bit of data lays in the interval $[m_j, M_j]$: $\beta_j = \delta_{T_j}$, with $T_j = M_j-m_j$ and we will add delay $m_j$ at the end of the analysis;
				\item[b)] $\beta_j$ is a strict service curve ;
				\item[c)] $\beta_j$ is a (min,plus) sub-additive service curve;
			\end{enumerate}
		%We denote $\cD$ the set of servers offering a delay service curve and $\cS$ the set of servers offering a non-delay service curve. Sets $\cD$ and $\cS$ form a partition of $\N_n$;
		\item we denote by  $A_i{(j)}$ the arrival process at server $j$ and $D_i^{(j)}$ is the departure cumulative process from server $j$, 
	\end{itemize}

	The processes above are defined considering they exist: $A_i^{(j)}$ is defined if $j\in\{ f_i, \ldots \ell_i+1\}$ ($A_i^{(\ell_i+1)}$ is the departure process from the system), and $D_i^{(j)}$ is defined if $j\in\{ f_i, \ldots \ell_i\}$, and for $j\in\{ f_i, \ldots \ell_i\}$, $A_i^{(j+1)} = D_i^{(j)}$.

	%Assume that flow 1 is the flow of interest and traverses the whole
	%tandem ($f_1 = 1$ and $\ell_1 = n$). 
	
	We denote $i\in j$ or $j\ni i$ if flow $i$ crosses server $j$.
	
	\subsection{PMOO analysis}
	Although the proof is similar to the ones in ~\cite{BGLT2007}, it is slightly gneralized, as taking more models of service curves into account. This is why we rewrite it here. 
	%The PMOO analysis have already been proved in several context (deterministic or stochastic network Calculus), so the proof presented here is not new. We nevertheless present it here, in this slightly generalized context (in particular it takes into account the delay curves.
	
	We perform the analysis backward from server $n$ to server 1. We start by the following lemma that is key for our backward induction. 

	\begin{lemma}
		\label{lem:pm}
		For all server $j$, for all $t_j\geq 0$, there exists $s_j$ such that 
		\begin{enumerate}
			\item $\sum _{i\in j}D_j^{(j)}(t_{j}) = \sum_{i\in j} A_i^{(j)}(s_j)+\beta_j(t_j-s_j)$.
			\item $\forall i\in j$, $D_i^{(j)}(t_j) \geq A_j^{(j)}(s_j)$. 
		\end{enumerate}
	\end{lemma}
	\begin{proof}
		Let $t_j\geq 0$.
		We disinguish the three possibilities for the service curves:
		\begin{itemize}
			\item[a)] If server $j$ is a delay service curve, then for all $i\in j$, $D_i^{(j)}(t_{j}) \geq A_i^{(j)}((t-T_j)_+)$, so  $$\sum_{i\in j} D_i^{(j)}(t_{j}) \geq \sum_{i\in j} A_i^{(j)}((t-T_j)_+) = \sum_{i\in j} A_i^{(j)}((t_j-T_j)_+) + \beta_j(\min(t_j, T_j)).$$
			\item[b)] If server $j$ offers a strict service curve $\beta_j$, then we set $s_j = start_j(t_j)$, the start of the backlogged period of $t_j$. 
			We know that $$\sum_{i\in j} D_i^{(j)}(t_j) \geq \sum_{i\in  j}A_i^{(j)}(s_j) + \beta_j(t_j-s_j),$$
			and that at $s_j$ (start of the backlog period of server $j$, $\sum_{i\in j}A_i^{(j)}(s_j) = \sum_{i\in j}D_i^{(j)}(s_j)$, and since for all $i\in j$, $A_i^{(j)}(s_j) \geq  D_i^{(j)}(s_j)$ (causality of the system), we necessarily have $D_i^{(j)}(s_j) = A_i^{(j)}(s_j) $, so $D_i^{(j)}(t_j)\geq D_i^{(j)}(s_j) = A_i^{(j)}(s_j)$.
			\item[c)] If server $n$ offers a sub-additive service curve, then we set $s_n =u_n(t_n)$ and the case is similar to the previous one: 
			We know that $$\sum_{i\in n} D_i^{(n)}(t_n) \geq \sum_{i\in n}A_i^{(n)}(s_n) + \beta_n(t_n-s_n),$$
			and $D_i^{(j)}(t_j)\geq D_i^{(j)}(s_j) = A_i^{(j)}(s_j)$ for all $i\in j$, since  $\sum_{i\in j}A_i^{(j)}(s_j) = \sum_{i\in j}D_i^{(j)}(s_j)$ (Lemma~\ref{lem:start}).
		\end{itemize}
	\end{proof}
	
	We now set $t_n\geq 0$, and use the previous lemma to define $s_n$. If $s_j$ is defined, than we set $t_{j-1} = s_j$, until $s_1$ is defined, and $t_0 = s_1$ ($t_j$ are then well-defined by a backward induction). 
	We can write for all $j\in\N_n$,
	
	$$\sum_{i\in j} D_i^{(j)}(t_j) \geq  \sum_{i\in j} A_i^{(j)}(t_{j-1}) + \beta_j(t_j-t_{j-1}).$$
	
	Let us sum these $n$ inequalities: 
	$$\sum_{j=1}^n \sum_{i\in j} D_i^{(j)}(t_j) \geq  \sum_{j=1}^n\Big(\sum_{i\in j} A_i^{(j)}(t_{j-1}) + \beta_j(t_j-t_{j-1})\Big).$$
	We can simplify this inequality since $D_i^{(j)} = A_i^{(j-1)}$ for all $i$ and $j\in \{f_i+1, \ldots, \ell_i-1\}$:
	
	 \begin{equation}
	 	\label{eq:pmoo}
	 	\sum_{i=1}^m D_i^{(\ell_i)}(t_{\ell_i}) \geq  \sum_{i=1}^m  A_i^{(f_i)}(t_{f_i-1}) +\sum_{j=1}^n\beta_j(t_j-t_{j-1}).
	 \end{equation} 
	 For all $i\neq 1$, since the system is causal, $D_i ^{(\ell_i)}\leq A_i^{(f_i)}$, so we can bound $$D_i ^{(\ell_i)}(t_{\ell_i}) - A_i^{(f_i)}(t_{f_i-1})\leq A_i ^{(\ell_i)}(t_{\ell_i}) - A_i^{(f_i)}(t_{f_i-1})\leq \alpha_i(t_{\ell_i} - t_{f_i-1}).$$
	 and  integrating in Equation~\eqref{eq:pmoo}, we obtain
	 
	 $$  D_1^{(n)}(t_{n}) \geq    A_1^{(1)}(t_0) +\sum_{j=1}^n\beta_j(t_j-t_{j-1})- \sum_{i=2}^m \alpha_i(t_{\ell_i} - t_{f_i-1}).$$
	 To conclude, it remain to notice that (from Lemma~\ref{lem:pm}) $$D_1^{(n)}(t_n)\geq A_1^{(n)}(t_{n-1}) = D_1^{(n-1)}(t_{n-1})\geq \cdots \geq A_1^{(1)}(t_0), $$
	 so 
	$$  D_1^{(n)}(t_{n}) \geq    A_1^{(1)}(t_0) +\Big(\sum_{j=1}^n\beta_j(t_j-t_{j-1})- \sum_{i=2}^m \alpha_i(t_{\ell_i} - t_{f_i-1})\Big)_+,$$
	and changing the notation using $u_j = t_j-t_{j-1}$ we have just proved the following  theorem.
	\begin{theorem}
		Consider a tandem network, and flow 1 traversing the $n$ servers. Then an end-to-end individual service curve  for flow 1 is 
		$$\beta_{e2e} (t)= \Big(\inf_{u_j\geq 0:~\sum_{j=1}^n u_j = t}\sum_{j=1}^n \beta(u_j) - \sum_{i=2}^m \alpha_i(\sum_{j\ni i} u_j)\Big)_+. $$
	\end{theorem}

	\subsection{Linear model: computation of the residual service curve}
	Computing individual service curves of the PMOO analysis is a NP-hard problem in general~\cite{BJT08}. However, whenarrival curves are
	leaky-bucket ($\gamma_{b_i, r_i}$ for flow $i$) and service curves
	rate-latency ($\beta_{R_j, L_j}$ for component $j$) writing an explicit formula is possible~\cite{BGLT2007}:
	 $\beta_{e2e}$ is a rate-latency service curve with rate $R_{e2e}$ and latency $T_{e2e}$:
	$$R_{e2e} = \min_j (R_{j}-\sum_{i\in j, i\neq 1}r_i)\qquad T_{e2e} = \sum_jT_j + \frac{\sum_i (b_i+ r_i\sum_{j\ni i} T_j )}{R_{e2e}}.$$
	%$$R_{e2e} = (R-r_{3} - r_2\qquad T_{e2e} = \sum_{j=1}^n T_j + \frac{b_2+\sum_{j=i}^n b_{2+j} +(r_{2+j}+r_2)T_j}{R_{e2e}}. $$

	\begin{example}
	
	In the case of interest of \cite{HCS24}, the system is a sequence of $n$  pure delay curves ($R_j = \infty$) followed by a sub-additive constant-rate service curves ($T_j = 0$), and the traffic beyond the flow of interest is one flow (flow 2) crossing the whole system (similar to flow 1) and $n$ flows (flow 3 to $n+1$) crossing each delay / constant rate subsystem. 
	We can compute
	
	$$R_{e2e} = \min_{j}(R_j-r_{j+2}) - r_2\qquad T_{e2e} = \sum_{j=1}^n T_j + \frac{b_2+\sum_{j=i}^n b_{2+j} +(r_{2+j}+r_2)T_j}{R_{e2e}}. $$
	Moreover, the numerical evaluation seems done in a homogeneous context: all servers have the same rate $R$ and pure delays $T$, $r_{2+j} = r_3$, and $b_{2+j} = b_3$ so this simplifies into
	$$R_{e2e} = R-r_2-r_3\qquad T_{e2e} = nT + \frac{b_2+n [b_{3} +  (r_{2}+r_3)T]}{R-r_2-r_3}, $$
	and the end-to-end delay is: $$D = nT + \frac{b_2+n [b_{3} +  (r_{2}+r_3)T]}{R-r_2-r_3} + \frac{b_1}{R-r_2-r_3}.$$
	This is to be compared with the delays obtained with the {\em unconventional analysis}: the end-to-end service curve can be obtained similarly, as there is no special attention  for the order of operations: 
	
	The residual service curve is 
	\begin{equation}
		\label{eq:jens}
		\beta_{nc}(t): t\mapsto \max \Big(-nb_3-b_2-(r_2-r_3)t, R(t-nT)-nb_3-b_2-(r_2-r_3)t).
	\end{equation}
	 
	We can interpret it as 
	\begin{enumerate}
		\item Compute the individual  service curve for flows 1 and 2 for each delay/rate system : $\beta'(t) = \max(-b_3-r_3t, R(t-T)-b_3-r_3t)$.
		\item Compute the end-to-end service curve for flows 1 and  2 (concatenating $n$ times  $\beta'$): $\beta_{1, 2} = \beta'^n = \max (-nb_3-r_3t, R(t-nT)-nb_3-r_3)t)$.
		\item Computing the individual service curve $\beta_{nc} = \beta_{1, 2} - \alpha_2$ leads to the formula of Equation~\eqref{eq:jens}. 
	\end{enumerate}
	
	the residual service curve computed from the service curve $\beta_{R, nT}$ (concatenation of the overall service curve) and the cross-flow with arrival curve $\gamma_{b_2+nb_3, r_2+r_3}$. We can then apply the formula from Equation~\eqref{eq:delay-nc} to compute the delay bound:
	$$D_{nc} = \max\Big(nT +\frac{b_2+nb_3 + (r_2+r_3)nT+b_1}{R-r_2-r_3}, nT + \tau_1 +\frac{b_2+nb_3 +(r_2+rR_3)nT}{r'_1}\Big),$$
	where $(r'_1(t-\tau_1)_+)$ is the minimum arrival curve for flow 1. One can observe that the first term of the maximum is the same as $D$, so the delay computed is always pessimistic compared to more conventional analysis. Second, since $r'_1\leq r_1<R-r_2-r_3$ (stability condition), then the second term is likely to be larger then the first.
	
	Using the numbers from \cite{HCS24}, we compare the two bounds: 
	\begin{itemize}
		\item $n$ varies from 1 to 20;
		\item $R = 20Mbps$, $T = 50 ms$;
		\item $r_i = 5Mbps$, $b = 1Mb$, $\tau_1 = 50ms $	;
		\item $r'_ 1 \in\{0.5, 1.25, 2.5, 3.75, 5\}Mbps$. 
	\end{itemize} 
	
	On Figure~\ref{fig:J1} one can verify that the PMOO approach indeed allows to compute tighter end-to-end delay curves, and the  non-conventional analysis worsen as the minimum rate guarantee decreases. 
		\begin{figure}[htbp]
			\centering
			\includegraphics[scale=0.5]{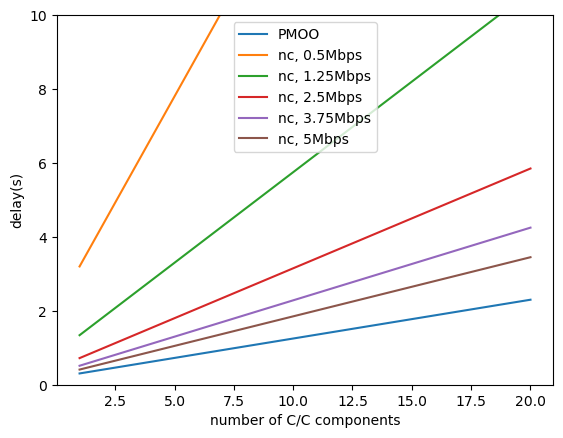}
			\caption{Tandem analysis: delay non-conventional analysis vs. improved PMOO analysis.}
				\label{fig:J1}
		\end{figure}
\end{example}
%	
%	First we can see that the infimum cannot be obtained if
%	$u_j < \delta_j + T_j$. In that case, we have, with $(*)$ being what
%	is inside the positive part,
%	\begin{align*}
%		(*) & = \inf_{\sum u_j = t-s, u_j \geq \delta_j + T_j}\sum_{j}R_j(u_j - T_j-\delta_j) - \sum_{i\neq 1} b_i + r_i(\sum_{j\ni i}u_j)\\
%		& = \inf_{\sum v_j = t-s - \sum_j(T_j + \delta_j), v_j \geq 0}\sum_{j}R_j(v_j) - \sum_{i\neq 1} b_i + r_i(\sum_{j\ni i}(v_j + T_j + \delta_j)\\
%		& = [\inf_{\sum v_j = t-s - \sum_j(T_j + \delta_j)} \sum_{j}(R_j-\sum_{i\in j, i\neq 1}r_i)v_j ]-\sum_i (b_i+ r_i\sum_{j\ni i} (T_j + \delta_j))
%	\end{align*}
%	Let $j_0$ a the server at which $R_j-\sum_{i\in j, i\neq 1}r_i$ is minimized. Then 
%	$$\inf_{\sum v_j = t-s - \sum_j(T_j + \delta_j)} \sum_{j}(R_j-\sum_{i\in j, i\neq 1}r_i)v_j  = (R_{j_0}-\sum_{i\in j_0, i\neq 1}r_i)[t-s - \sum_j(T_j + \delta_j)],$$
%	and
%	
%	$$(*) = (R_{j_0}-\sum_{i\in j_0, i\neq 1}r_i)[t-s - \sum_j(T_j + \delta_j)] - \sum_i (b_i+ r_i\sum_{j\ni i} (T_j + \delta_j))$$
%	
%	Finally, a residual service curve for flow 1 is a rate-latency
%	function with rate
%	$$R^{res} = \min_j (R_{j}-\sum_{i\in j, i\neq 1}r_i)$$
%	and latency
%	$$T^{res} = \sum_j(T_j + \delta_j) + \frac{\sum_i (b_i+ r_i\sum_{j\ni i} (T_j + \delta_j))}{R^{res}}.$$
%	
%	

	\section{Use-case 2: Feedback  and window-flow control}
	\label{sec:wfc}
Feedback loops have been defined for admission control and to prevent buffer overflows  in a network. More precisely, it controls the amount of data in a part of the system. If the maximum amount of data allowed in the system is $W$, then data is prevented from entering the system when the amount of data is already $W$. Data is then buffered at the entrance, and admitted only as data exits the system. This is modeled by a feedback loop. 

\subsection{A simple system with feedback control}
The simpler feedback-loop  system is depicted in Figure~\ref{fig:feedback}: the arrival process $A$ enters the system. The flow crosses the network element offering a service curve $\beta$. There is a feedback control: the arrival process $A$ is limited by the process $D = B+W$, and $A' = A\oplus D$ is the arrival process to the server. 
The precise description is given below, where we consider linear arrival and service curves for the sake of simplicity. 

\begin{itemize}
	\item The arrival process to the system $A$ is $\alpha$-constrained, with $\alpha = \gamma_{r, b}$;%:t\mapsto \sigma+\rho t$; 
	\item There is one server offering a with (min,plus) or strict service curve $\beta = \beta_{R, T}$; %:t\mapsto R(t-T)_+$; 
	\item the feedback control is a fixed window control $W$, ensuring the backlog in the server is never more than $W$ (an additional queue is added at the entrance of the system to prevent losses).
\end{itemize}

 %and simplify the problem to avoid working with complex functions and computations. This is done at the price of degrading the performance bounds of the system (it might be a little more pessimistic than with more complex computation that can be computed with Nancy for example. 

\begin{figure}[htbp]
	\centering
	\begin{tikzpicture}
		\node[rectangle, draw, minimum width=0.5cm](fb) at (0, 0) {$W$};
		\node[rectangle, draw, minimum width=1cm](s) at (0, -2) {$\beta$};
		\node[circle](mp) at (-2, -2)  {\Huge{$\oplus$}};
		\draw[->] (-1.7, -2) -- (s) node [pos=0.5, above] {$A'$};
		\draw[->](-4, -2)-- (-2.3, -2) node [pos=0.8, above] {$A$};
		\draw[->] (s) -- (2, -2)  node [pos=0.3, above] {$B$}-- (2, 0) -- (fb);
		\draw[->] (fb)-- (-2, 0) -- (-2, -1.7)node [pos=0.7, left] {$D$}; 
		\draw[->] (s) -- (3, -2);
	\end{tikzpicture}
	\caption{An elementary feedback loop.}
	\label{fig:feedback}
\end{figure}
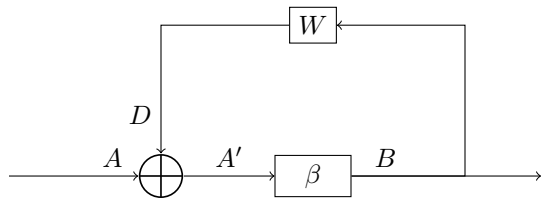

We can write the system of (in)equalities: 
$$\left\{\begin{array}{ll}
	A' &= A\land D\\
	D  & = B+W = B\conv \varphi_W\\
	B &\geq A'\conv \beta,
\end{array}\right.$$
with $\varphi_W:\R_+\to\Rmin; t\mapsto W$. In fact, $\varphi_W$ could be replaced by any function $\varphi$, as soon as $\varphi(0)>0$.  

%$A' = A\oplus D$, with $D = B+W = B\conv \varphi_W$, and $B\geq A'\conv \beta$. 
Putting everything together, we obtain $A'\geq A\land A'\conv \beta \conv \varphi_W$, and the smallest solution to this equation is $A'\geq A\conv (\beta \conv \varphi_W)^*$. 

In our setting, $\beta\conv \varphi_W(t) = W + R(t-T)_+$, and $(\beta\conv \varphi_W)^*\geq\beta'$ with 
\begin{equation}
	\label{eq:throttle}
	\beta_{th}:0\mapsto0;~t\mapsto  W +\min (R, \frac{W}{T} )(t-T)_+.
\end{equation}
%$$\beta':0\mapsto0;~t\mapsto  W +\min (R, \frac{W}{T} )(t-T)_+.$$ 
Here, to simplify the expression, we only give an lower bound of the exact computation, detailed in Figure~\ref{fig:ex}. 

	\begin{figure}[htbp]
	\centering
	\begin{tikzpicture}[scale=0.8]
		\draw[->] (0, 0) -- (7, 0) node[below, pos = 0] {0} node[below, pos = 1] {time};
		\draw[->] (0, 0)-- (0, 5) node[left, pos=0.9] {\rotatebox{90}{data}};
		%\draw[red,  thick](0, 1) -- (2, 1) node[pos=0, left] {$W$} -- (7, 3) ;
		\draw[blue, ultra thick, dashed](0, 2) -- (4, 2) node[pos=0, left] {$2W$} -- (7, 3.2);
		\draw[olive, ultra thick] (0, 0) -- (0, 1)-- (2, 1)  -- (7, 3) ;
		\draw[red](0, 1) -- (2, 1) node[pos=0, left] {$W$} -- (7, 3) ;
		%\draw[cyan, ultra thick, dashed] (0, 0)--(2, 1) ;
		%\draw[cyan, ultra thick] (2, 1) -- (7, 3.5) ;
		
		\draw[red, ultra thick] (0.5, 4.5) -- (1, 4.5) node[pos=1, right]{$\beta\conv \varphi_W$};
		\draw[blue, ultra thick, dashed](0.5, 4) -- (1, 4) node[pos=1, right]{$(\beta\conv \varphi_W)^2$};
		\draw[olive, ultra thick](0.5, 3.5) -- (1, 3.5) node[pos=1, right]{$(\beta\conv \varphi_W)^*$};
		%\draw[cyan, ultra thick](0.5, 3) -- (1, 3)node[pos=1, right]{$\beta'\leq (\beta\conv \varphi_W)^*$};
	\end{tikzpicture}
	\begin{tikzpicture}[scale=0.8]
		\draw[->] (0, 0) -- (7, 0) node[below, pos = 0] {0} node[below, pos = 1] {time};
		\draw[->] (0, 0)-- (0, 5) node[left, pos=0.9] {\rotatebox{90}{data}};
		\draw[red, ultra thick](0, 1) -- (2, 1) node[pos=0, left] {$W$} -- (5, 4) ;
		\draw[blue, ultra thick, dashed](0, 2) -- (4, 2) node[pos=0, left] {$2W$} -- (6, 4);
		\draw[cyan, ultra thick, dashed] (0, 0)--(2, 1) ;
		\draw[cyan, ultra thick] (2, 1) -- (7, 3.5) ;
		\draw[olive, ultra thick] (0, 0) -- (0, 1)-- (2, 1) -- (3, 2) -- (4, 2) -- (5, 3) -- (6, 3) -- (7, 4) ;
		\draw[red, ultra thick] (0.5, 4.5) -- (1, 4.5) node[pos=1, right]{$\beta\conv \varphi_W$};
		\draw[blue, ultra thick, dashed](0.5, 4) -- (1, 4) node[pos=1, right]{$(\beta\conv \varphi_W)^2$};
		\draw[olive, ultra thick](0.5, 3.5) -- (1, 3.5) node[pos=1, right]{$(\beta\conv \varphi_W)^*$};
		\draw[cyan, ultra thick](0.5, 3) -- (1, 3)node[pos=1, right]{$\beta_{th}\leq (\beta\conv \varphi_W)^*$};
	\end{tikzpicture}
	\caption{Example of computation and approximation of $(\beta \conv \varphi_W)^*)$ when $R>\frac{W}{T}$. Left: $RT\leq W$, $\beta\conv \varphi_W$ and $(\beta\conv \varphi_W)^*$ differ only at 0; right:$RT>W$$(\beta\conv \varphi_W)^*$: the $(\beta\conv \varphi_W)^n$ are incomparable, and we obtain an function that can be approximated by a nearly linear curve.}
	\label{fig:ex}
\end{figure}
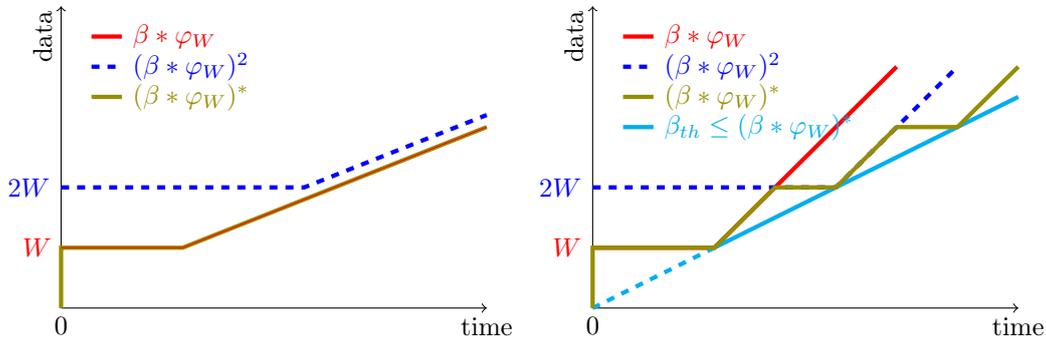

Similar to \cite{ACOR99}, one can interpret $(\beta \conv \varphi_W)^*$ as a service curve (\cite[th. 18]{ACOR99}, the {\em throttle service curve}). One can remark that this service curve is by construction sub-additive.

\paragraph{Interpretation of the feedback loop as a server.} One interesting fact in the above analysis is the obtention of the equation $A'\geq A\conv (\beta \conv \varphi_W)^*$,that can directly be interpreted as an arrival/departure relation, $A$ being the arrival process, and $A'$, the departure process of a server offering (min,plus) the service curve  $(\beta \conv \varphi_W)^*$. The overall system of Figure~\ref{fig:feedback} can then be replaced by the two-server system if Figure~\ref{fig:open}. 

%We can see that $\beta'(0) = 0$, so it fulfills the usual hypothesis for a service curve. This is a priori a (min,plus) service curve. 
%\begin{example}
%	Still assuming that $A$ is $\alpha$-constrained, with de service curve found in the previous example
%\end{example}

%and in this specific case, since $\beta'$ is discontinuous at 0, it might be possible to compute a {\em maximum backlogged period length}, hence worst-case delays without having to resort to service curves with negative values. 

The equivalent system if shown in Figure~\ref{fig:open}.

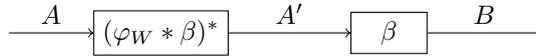
\begin{figure}[htbp]
	\centering
	\begin{tikzpicture}
		\node[rectangle, draw, minimum width=0.5cm](fb) at (-3, -2) {$(\varphi_W*\beta)^*$};
		\node[rectangle, draw, minimum width=1cm](s) at (0, -2) {$\beta$};
		
		\draw[->] (-5, -2) -- (fb) node [pos=0.5, above] {$A$};
		\draw[->](fb)-- (s) node [pos=0.5, above] {$A'$};
		\draw[->] (s) -- (2, -2)  node [pos=0.5, above] {$B$};
		
	\end{tikzpicture}
	\caption{Interpretation of the system of Fig.~\ref{fig:feedback} as an open-system.}
	\label{fig:open}
\end{figure}

Let us compute the throttling service curve:

\begin{example}
	If $R\leq\frac{W}{T}$, then $(\beta \conv \varphi_W)^*(t)  = W+ R(t-T)_+$ for all $t>0$, and $\beta\conv (\beta \conv \varphi_W)^*(t)  = \beta$: the throttle  service curve has no effect on the system. 
	
	If $R>\frac{W}{T}$, then $(\beta \conv \varphi_W)^*(t)  \geq  W+ \frac{W}{T}(t-T)_+ \geq  \frac{W}{T}t$ for all $t>0$, and $\beta\conv (\beta \conv \varphi_W)^*(t)  \geq \frac{W}{T}(t-T)_+$. Figure~\ref{fig:ex} shows the computation and the approximation made when computing $(\beta \conv \varphi_W)^*)$. In this case, the throttle server has en effect on the computation of the end-to-end service curve of the system. The stability condition is even modified: if the arrival process is $\gamma_{r, b}$-constrained, the service rate if the throttle server is $\frac{W}{T}$, the stability condition becomes $r<\frac{W}{T}$. 
\end{example}
	
	Note that if the server offers a strict service curve~\cite{HWS25} show that the throttle has no effect in the system, so there is no additional stability condition. However, when the server is replaced by $k$ servers in tandems, as the end-to-end service curve can only be (min,plus), the throttle server has effect whatever the type of service curve.

We are now interested in the computation of the arrival curve of $A'$, the departure process from the throttle server.

\begin{example}
	\label{ex:output}
	Still assuming that $A$ is $\gamma_{r, b}$-constrained, with the throttle service curve computed Equation~\ref{eq:throttle} $\beta_{th}(t) = \min (W/T t,W+R(t-T)_+)$. We assume stability, $r<W/T$. Then an arrival curve for $A'$ is $\alpha' = \gamma_{r, b}\deconv\beta_{th}$. If $\frac{W}{T}<R$, we get 
	$\alpha'(t) =\sup_{u\geq 0} b+r(t+u) - \frac{W}{T}u =\gamma_{r, b}(t)$ (using the stability hypothesis). If $\frac{W}{T}\geq R$, we get for $t>0$,  $\alpha'(t) = \sup_{u\geq 0} b +r(t+u) - W-R(t-T)_+ = \max(b+rt, b + rt -W+RT) = \gamma_{r, b}$, since $r<R<W/T$. In both cases, the arrival curve in unchanged and the throttling does not create any increase in the arrival curve.  
\end{example}

Alternatively, in~\cite{HCS24}, the strategy is to consider this computation as a first step to computing the (min,plus) service curve of the whole system: $B\geq A\conv\beta\conv(\beta \conv \varphi_W)^*$. 

More preciely, the authors consider the case of $\beta$ being the concatenation of two a constant-rate servers, with respective rates $R_1$ and $R_2$, and $A$ being the superposition of two flows $A = A_H+A_L$, with high ($A_H$) and low ($A_L$) priority.

They obtain the end-to-end service curve for the aggregate flows:
$$
\beta_{e2e} = \lambda_{R_1}\conv \lambda_{R_2} \conv (\lambda_{R_1}\conv \lambda_{R_2} +B)^* = \lambda_{R_1}\conv \lambda_{R_2} = \lambda_{R_1\land R_2}.$$
The last equality is because $\lambda_{R_1}\conv \lambda_{R_2}$ and $(\lambda_{R_1}\conv \lambda_{R_2} +B)^*$ are sub-additive and $\lambda_{R_1}\conv \lambda_{R_2}\leq (\lambda_{R_1}\conv \lambda_{R_2} +B)^*$, and the individual service curve for the low priority flow is with the unconventional analysis:
$$\beta_{e2e, L}^{nc} =\beta_{e2e} - \alpha_H.$$ 

However, for the conventional analysis, the authors first perform the {\em conventional analysis} for the hight priority flow (obtaining  the  service curves for $A^L$ for the two servers, then perform another window-flow control analysis for the lower priority flow with a window size $B-B_{max}^H$: the window is reduced by the maximum occupancy of the high priority flow. It can be noticed that this line of reasoning is  very pessimistic, since a) the amount of data of the high priority flow in the buffer is not constantly $B_{\max}^H$ (this is a worst case that could happen very infrequently), and b) the analysis could be made directly with the aggregation of the two flows in the original system. The servers being constant rate means that the window-control  is in fact  useless in that case. 

Using our slightly improved conventional analysis, analyzing this system is in fact straightforward: one can perform the PMOO analysis of the system, since all service curves are sub-additive, and the end-to-end (min,plus) service curve for the lower-prority flow is 

$$\beta_{e2e, L} = (\beta_{e2e} - \alpha_H)_+ = (\lambda_{\min(R_1, R_2}-\alpha_H)_+.$$

%In addition, if the new objective was to compute the end-to-end delay, one could take advantage of the sub-additivity of all the service curves to perform the PMOO analysis even if all service curves are (min,plus). one first compute the throttle service curve 
%Alternatively, if the delay now matters one could also use our analysis: the open-loop system, with throttle server is represented in Figure \ref{fig:uc2}, and both flows cross the entire system. The PMOO analysis provides the individual service curve $$\beta_{e2e, L} = (\beta_{e2e} - \alpha_H)_+ = (\lambda_{\min(R_1, R_2}-\alpha_H)_+.$$

We obtain a larger service curve (non-negative), hence the delay bound computed with the conventional analysis is more accurate, and does not require a minimal arrival curve.

These results are valid for a single flow, or for several flows that all cross the whole system. In the other paper~\cite{HWS25}, the authors claim they can generalize the unconventional analysis for complex feedback structures. This the topic of the next section.

\section{Tandem network with complex feedback structures}
\label{sec:cplx}
A tandem network with complex feedback structures is a tandem network where several window control are intricate. More precisely, the tandem network is a sequence of $n$  servers with respective service curves $\beta_j$, $j\in\N_n$, and $k$ feedback loops represented by $k$ triples $t_i = (s_i, u_i, W_i)$, where $s_i\leq u_i$ are the first and last servers inside the feedback loop, and $W_i$ is the maximum amount of data in the servers between $s_i$ and $u_i$. We will first derive a general end-to-end service curve for the system, in the case of a single flow. Then we explain how stability issues can arise when we consider multiple flows. We finally give some cases where these issues can be solved. % and where efficient computations can be performed. 
Complex strategies for solving more general cases are out of the scope of this paper. 

\subsection{Equivalent system in a single flow tandem}

In this paragraph, we consider that only one flow circulates inside the network: this flow traverses the $n$ servers. We will derive an equivalent open-loop system, where, following the analysis of the previous section, we add one throttle server per feedback loop. With $k$ throttle servers in total, we will obtain a tandem network with $n+k$ server.

 Let us first discuss the place of the throttle servers. In general a triple $t = (a, u, W)$ is placed at the synchronization point, just before server $s$. However, it might be the case that several triples $t_1 = (s_1, u_1, W_1)$ and $t_2 = (s_2,u_2, W_2)$ have the same {\em starting server} ($s_1 = s_2)$, and synchronize at the same point. In that case, we set the convention that the throttle server for $t_1$ is places before that of $t_2$ if $u_1>u_2$ (without loss of generality, we assume  that  $u_1\neq u_2$, otherwise one could keep only the triple $t_i$ with the smallest $W_i$).

\begin{example}
	Consider the example in Figure~\ref{fig:cpx-ex}, with three servers and three feedback loops. The three triples are $t_1 = (1, 3, W_{1})$, $t_2 = (1, 2, W_{2})$ and $t_3 = (2, 3, W_{3})$. There are two throttle servers $\Psi_1$ and $\Psi_3$ to place before server 1, and $\Psi_3$ is placed  before server 2. From the  convention above, $\Psi_1$ is placed before $\Psi_3$. We then obtain the 6-server tandem as in Figure~\ref{fig:net-cex-open}.
	
	\begin{figure}[htbp]
		\centering
		\begin{tikzpicture}
			\node[rectangle, draw, minimum width=0.5cm](fb1) at (0, 0) {$W_{2}$};
			\node[rectangle, draw, minimum width=0.5cm](fb2) at (3,  -1) {$W_{3}$};
			\node[rectangle, draw, minimum width=0.5cm](fb3) at (3,  -3) {$W_{1}$};
			\node[rectangle, draw, minimum width=1cm](s1) at (0, -2) {$\beta_1$};
			\node[rectangle, draw, minimum width=1cm](s2) at (3, -2) {$\beta_2$};
			\node[rectangle, draw, minimum width=1cm](s3) at (5, -2) {$\beta_3$};
			\node[circle](mp1) at (-1.5, -2)  {{\huge $\oplus$}};
			\node[circle](mp2) at (1.5, -2)  {{\huge $\oplus$}};
			\draw (-1.5, -2) -- (s1) -- (s2) -- (s3) ;

			\draw[->] (s2) -- (4, -2)   -- (4, 0) -- (fb1) -- (-1.5, 0) -- (-1.5,-1.75) ;
			\draw[->] (s3) -- (5.7, -2)   -- (5.7, -1) -- (fb2) -- (1.5, -1) -- (1.5,-1.75) ;
			\draw[->] (s3) -- (5.7, -2)   -- (5.7, -3) -- (fb3) -- (-1.5, -3) -- (-1.5,-2.2) ;
			\draw[red, ultra thick, ->] (-2.5, -2.1) -- (6.5, -2.1) node[pos=0, below] {$A_0$} node[pos=0.2, below] {$A_2$} node[pos=0.4, below] {$A_3$} node[pos=0.5, below] {$A_4$} node[pos=0.7, below] {$A_5$}node[pos=0.95, below] {$A_6$} ;
		\end{tikzpicture}
		\caption{Tandem network with a complex feedback structure.  }
		\label{fig:cpx-ex}
	\end{figure}
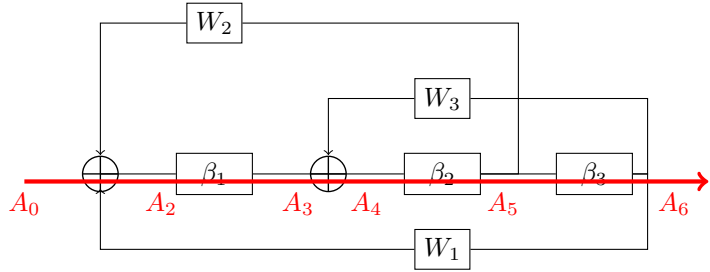
	
	\begin{figure}[htbp]
		\centering
		\begin{tikzpicture}
			\node[rectangle, draw, minimum width=0.5cm](fb1) at (0, 0) {$\Psi_1$};
			\node[rectangle, draw, minimum width=0.5cm](fb2) at (1.5,  0) {$\Psi_2$};
			\node[rectangle, draw, minimum width=0.5cm](fb3) at (4.5,  0) {$\Psi_3$};
			\node[rectangle, draw, minimum width=1cm](s1) at (3, 0) {$\beta_1$};
			\node[rectangle, draw, minimum width=1cm](s2) at (6, 0) {$\beta_2$};
			\node[rectangle, draw, minimum width=1cm](s3) at (7.5, 0) {$\beta_3$};

			\draw[red, ultra thick, ->] (-1, -0.1) -- (8.5, -0.1) node[pos=0, below] {$A_0$} node[pos=0.2, below] {$A_1$}  node[pos=0.35, below] {$A_2$} node[pos=0.53, below] {$A_3$} node[pos=0.65, below] {$A_4$} node[pos=0.83, below] {$A_5$}node[pos=1, below] {$A_6$}  ;
		\end{tikzpicture}
		\caption{Open-loop version of the network of Figure~\ref{fig:cpx-ex}. }
		\label{fig:net-cex-open}
	\end{figure}
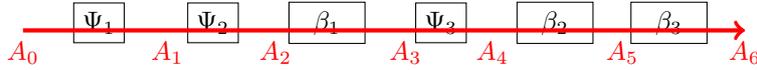
\end{example}

Our goal is now to compute the service curve $\Psi_i$ of the throttle server represented by triple $t_i$, for $i\in\N_k$.

%Assume there are $k$ feed-back loops, represented by the triple $t = (s, u, \varphi)$. In the single feed-back case, we add a throttle server at the synchronization point. Here we so the same, and will obtain a sequence of $n+k$ servers. 

%A throttle server is placed at the synchronization point (if the triple is $t = (s, u, W)$, then the corresponding throttle server is placed just before server $s$). If there are two triples with the triples $ t = (s, u, W)$ and $t' = (s, u', W')$, then the throttle server for $t$ is placed before that for $t'$ is $u>u'$.  
%Globally, for a tandem network with $n$ servers and $k$ feed-back control loops, we will compute an open system with $n+k$ servers, and $k$ service curves for the throttle servers need to be computed. 

Let us assume that the servers $\Psi_1,\ldots, \Psi_k$ are ordered from left to right in the tandem, for simplicity, we identify the server with its service curve ($\Psi_i$ or $\beta_j$). Let us define the arrival processes in system with the feedback loops: 
\begin{itemize}
	\item $A_0$ is the arrival process in the tandem network; 
	\item assume that $A_p$ is defined, $p\in\N$:
	\begin{itemize}
		\item if $A_p$ is at the synchronization of a feedback loop $i$, and there are $\ell$ triples that synchronize at this point,  we define $A_{p+\ell}$ as the process after the synchronization (we will write the equations later): more precisely, if $A_p$ is the arrival process at a synchronization point for triples $t_i, \ldots, t_{i+\ell-1}$, and that $A_{p_1}, \ldots, A_{p_{\ell}}$ are the departure processes after servers $u_i, \ldots, u_{i+\ell-1}$, then the have the relation
		\begin{equation}
			\label{eq:bf1}
			A_{i+\ell} = A_i\wedge\bigwedge_{r=1}^{\ell} A_{p_r}\conv \varphi_{W_{i+r-1}};
		\end{equation}
		%$$A_{i+\ell} = A_i\wedge\bigwedge_{r=1}^{\ell} A_{p_r}\conv \varphi_{W_{i+r-1}}.$$

		\item if $A_p$ is at the entrance of server $\beta_j$, then we set $A_{p+1}$ as the departure process after server $\beta_j$: if $A_p$ is the arrival process of server $j$, then
		\begin{equation}
			\label{eq:bf2}
			A_{p+1}\geq A_p\conv\beta_j;
		\end{equation}
	\end{itemize}
	\item the departure process of the system is $A_{n+k}$;
	\item we denote by $I$ the set of indexes $p$ for which $A_p$ is defined, and we define $\cS_{FB} = (A_p)_{p\in I}$ the set of cumultive processes, that satisfy Equations~\eqref{eq:bf1} and~\eqref{eq:bf2}. 
\end{itemize}

Let us define the arrival processes in system with the throttle servers: 
\begin{itemize}
	\item $A_0$ is the arrival process in the tandem network; 
	\item assume that $A_p$ is defined, $p\in\N$ at the arrival of a server. Then, $A_{p+1}$ is the departure process from that server: 
	\begin{equation}
		\label{eq:ol}
		A_{p+1}\geq A_p\conv\beta_j \text{ or }A_{p+1}\geq A_p\conv\Psi_{\ell},
	\end{equation}
	depending on the case. 
	
	\item $A_{p+n}$ is the departure process of this system;
	\item Similar to $\cS_{FB}$, we define $\cS_{OL} = (B_p)_{p\in I}$ the set of cumulative processes, that satisfy Equation~\eqref{eq:ol}. 
\end{itemize}

%Let us now write the equations for the feed-back control system. Assume that $A_p$ is defined:
%\begin{itemize}
%	%\item if $A_p$ is the arrival process of server $j$, then $A_{p+1}\geq A_p\conv\beta_j$;
%	\item if $A_p$ is the arrival process at a synchronization point for triples $t_i, \ldots, t{i+\ell-1}$, and that $A_{p_1}, \ldots, A_{p_{\ell}}$ are the departure processes after servers $u_i, \ldots, u_{i+\ell-1}$, then the have the relation
%	$$A_{i+\ell} = A_i\wedge\bigwedge_{r=1}^{\ell} A_{p_r}\conv \varphi_{W_{i+r-1}}.$$
%\end{itemize}

\begin{example}
	The processes for the tandem of the previous example are depicted on Figures~\ref{fig:cpx-ex} and~\ref{fig:net-cex-open}. 
There are two  windows synchronizing at $A_0$: $t_1 = (1, 3, W_{1, 3})$ and $t_2 = (1, 2, W_{1, 2})$. The departure processes after servers 2 and 3 are $A_5$ and $A_6$ , so we have the relation $A_2 = A_0\wedge A_5\conv\varphi_{W_{1, 2}}\wedge A_6\conv\varphi_{W_{1, 3}}$. 
\end{example}

%We denote $\cS_{FB} = \{(A_0, A_{n+k})\}$ the set of pairs of cumulative processes such that there exists $(A_0, \ldots, A_{n+k})$ that satisfy the equalities and inequalities of the feed-back control system. In short, these are the admissible arrival and departure processes of the feed-back loop system. Similarly, define $\cS_{OL} = \{(B_0, B_{n+k})\}$ as the set of trajectories so that there exists $(B_0, \ldots, B_{n+k})$ that satisfy the  inequalities of the open-loop  system. 

Our goal is to find service curves of the throttle server such that $\cS_{FB}\subseteq\cS_{OL}$. 
%Consider the $k$ triples $t_i = (s_i, u_i, W_i)$. Let us first define a partial order on the windows:
%$(s_1, t_1, W_1)\prec(s_2, t_2, W_2)$ if and only if $s_1<s_2$. The triples having tha same first component are not ordered.  We assume that $t_1\nsucc t_2\nsucc\cdots \nsucc t_k$. The placement of the throttle servers is then from left to right.  We say $t'\in I(t)$ if $s'<s\leq u'$ .

%Let $A_0$ be the arrival process at the network, and let $A_i$ the departure process from the $i$-th server (throttle or server).  In our original system, we have the following relations: 
%\begin{itemize}
%	\item if server $i$ corresponds to server $j$ in the original network, then $A_i\geq A_{i+1}\conv \beta_j$;
%	\item if server $i$ is the throttle server of $t_j$, then $A_i = A_{i-1}\oplus A_{u_j} \conv \varphi_{W_j}$. 
%\end{itemize}

 %We are now ready to compute the throttle  service curves. 

\newcommand{\bigconv}{\mathlarger{\mathlarger{\mathlarger{\ast}}}{}}

\begin{theorem}
	If the throttle service curves are defined for all $i\in\N_k$, 
	
	$$\Psi_{i} = \Big(\bigconv_{j=s_i}^{u_i } \beta_i \conv\big(\bigconv_{j~|~s_i< s_j\leq u_i}\Psi_{j}\big)\conv\varphi_{W_i}\Big)^*,$$
	then $\bigconv_{j=1}^n\beta_j\conv\bigconv_{\ell = 1}^k \Psi_{\ell}$ is an end-to-end service curve for the system. 
\end{theorem}
\begin{proof}
	We proceed by induction. The induction hypothesis $H(k)$ is the statement of the theorem if there are $k$ feedback loops. The case $H(0)$ is trivial, since the network is open-loop in the absence of feedback loop, then $\cS_{FB}=\cS_{OL}$ (the two systems satisfy the same inequalities). 
	
	Assume that the $H(k')$ is true for all $k'<k$, and assume a network with $k$ feedback loops. Consider $t_\ell,\ldots, t_k$ the triples such that $s = s_k =\cdots = s_{\ell}$ the triples with the largest first component (since triples are ordered, for all $j<\ell$, $s_{j}<s_{\ell}$). The throttle server of $t_j$, $j\in\{\ell, \dots,k\}$ is at server $s+j-1$. 
	
	Let us analyze the window-flow control and assume that the synchronization are at the same point (we skip the processes $A_{s+\ell -k-i}$ for $k\in\{\ell, k-1\}$, and reintroduce them later). 
	At the synchronization, one must have (Equation~\eqref{eq:bf1}) $$A_{k+s}  =  A_{s+\ell - 1} \land A_{s+u_\ell}\conv \varphi_{W_{\ell}}\land\cdots\land A_{s+u_k}\conv \varphi_{W_k}, $$ and we have for $j\in\{\ell, \ldots, k\}$ (Equation~\eqref{eq:bf2}), 
	$A_{s+t_j}\geq  A_{s+k-1}\conv \beta_{s}\conv\cdots\conv \beta_{u_j}$, so 
	
	$$A_{k+s-1}  =  A_{s+\ell - 1} \wedge A_{s+k-1}\conv \bigwedge_{j=\ell}^{k}  \bigconv_{i=s}^{u_j}\beta_{i}\conv\varphi_{W_j}, $$
	and $$A_{k+s} \geq A_{s+\ell - 1}\Big(\bigwedge_{j=\ell}^{k}  \bigconv_{i=s}^{u_j}\beta_{i}\conv\varphi_{W_j}\Big) ^* =  A_{s+\ell - 1}\bigconv_{j=\ell}^{k} \Big(  \bigconv_{i=s}^{u_j}\beta_{i}\conv\varphi_{W_j}\Big) ^*.$$
	We can set the throttle service curve for triple $t_j$, $j\in\{\ell, \ldots, k\}$, $$\Psi_j = \Big(  \bigconv_{i=s}^{u_j}\beta_{i}\conv\varphi_{W_j}\Big) ^*,$$ and 
	$A_{s+j} \geq A_{s+j-1} \conv \Psi_j$, such that $A_{k+s}\geq A_{s+\ell - 1} \conv\Psi_\ell \conv\cdots\conv\Psi_k$, as desired.

	Doing this operations removes at least one window-control, and we replaced each window by a throttle server, we obtain a tandem network with $n+k-\ell$ servers, and $k-\ell$ windows. We can now apply $H(k-\ell)$ to conclude. 
%	Since the triples are ordered, $t_k$ is the largest triple, and the throttle server is server $p = s_k-1+k$. We have for $i\in \{p+1, n\}$, $A_{i-1}\geq A_{i}\conv \beta_{i-k}$ and $A_p\geq A_{p-1}\oplus A_{p+u_k-s_k}\conv \varphi_{W_k}$. Following the single-feedback loop case, we have 
%	$$A_p\geq A_{p-1}\oplus A_{p}\conv \Big(\bigconv_{j=s_k}^{u_k} \beta_j\Big)\conv \varphi_{W_k},$$
%	and 
%	$$A_p\geq A_{p-1}\conv \Big[\Big(\bigconv_{j=s_k}^{u_k} \beta_j\Big)\conv \varphi_{W_k}\Big]^*,$$
%	and we can define the throttle service curve for the $k$-th feed-back loop $$\Psi_k = \Big[\Big(\bigconv_{j=s_k}^{u_k} \beta_j\Big)\conv \varphi_{W_k}\Big]^*.$$
%	Since	$I(t_k)\cap\cT = \emptyset$, this is exactly the service curve given in the theorem. After doing this operation, we reduce the number of feed-back loops by one and can apply the  induction hypothesis $H(k-1)$. Server $p$ is now seen as a first-type server. and has to be accounted fo the windows containing it, that the windows $t_i$ such that $t_k\in I(t_i)$. One can compute: 
\end{proof}

\begin{example}
	Consider the example in Figure~\ref{fig:cpx-ex}, with three servers and three feedback loops. The three triples are $t_1 = (1, 3, W_{1, 3})$, $t_2 = (1, 2, W_{1, 2})$ and $t_3 = (2, 3, W_{2, 3})$, and we have $t_1\prec t_2\prec t_3$. 
	We want to transform this tandem into a 6-server tandem, as in Figure~\ref{fig:net-cex-open}.
	
	\begin{enumerate}
		\item $\Psi_3 = (\beta_2\conv\beta_3\conv\varphi_{W_{3}})^*$;
		\item $\Psi_2 = (\beta_1\conv\beta_2\conv \Psi_3\conv \varphi_{W_{2}})^* = (\beta_1\conv\beta_2\conv (\beta_2\conv\beta_3\conv\varphi_{W_{3}})^*\conv \varphi_{W_{2}})^*$;
		\item $\Psi_1 = (\beta_1\conv\beta_2\conv \beta_3\conv \Psi_3\conv \varphi_{W_{1}})^* = (\beta_1\conv\beta_2\conv \beta_3\conv (\beta_2\conv\beta_3\conv\varphi_{W_{3}})^*\conv \varphi_{W_{1}})^*$. 
	\end{enumerate}
\end{example}

\subsection{The problem with complex structure}

In~\cite{HWS25}, the authors derive equivalent service curves, the same as in our theorem, but without a proof nor an explicit formula, in the case of a single flow analysis. This allows breaking the feedback loops, and considering an open-loop system, that can be analyzed it with conventional analysis.  A key difference of the two formulation is that the authors of~\cite{HWS25} do not use throttle server, and rather derive a new service curve per server (in our exampled, one would have $\beta'_1 = \Psi_1\conv\Psi_2\conv \beta_1$, $\beta'_2 = \Psi_3\conv \beta_2$ and $\beta'_3 = \beta_3$).%(the overall service curve being the convolution of all service curves). 

In the second step of their analysis, the authors of ~\cite{HWS25} use this single flow analysis to analyze the network when it is crossed by multiple flows. In this paragraph, we claim and demonstrate that this open-loop representation is not a safe model for the multi-flow analysis, that can lead to instability issues. 

To illustrate this, let us consider the simplified topology of ~\cite{HWS25} presented in Section~5 and Figures 11 and 12. This network is composed of 5 flows and 6 servers. In particular, there is a feedback-control  including servers 3 and 4, the goal being to control the buffer occupancy of server 4. Flow 5 arrives at server 2 and crosses servers 2 and 3, flow 4 arrives at server 4 (and crosses servers 4 and 5).  We can simplify this network as a network with 2 servers and 2 flows, and a feedback loop around the two servers  as depicted in Figure~\ref{fig:net-cex}. The important fact is that {\bf the control is not made by the entering flows}. 

\begin{figure}[htbp]
	\centering
	\begin{tikzpicture}
		\node[rectangle, draw, minimum width=0.5cm](fb) at (0, 0) {$W$};
		\node[rectangle, draw, minimum width=1cm](s1) at (0, -2) {$\beta_1$};
		\node[rectangle, draw, minimum width=1cm](s2) at (2, -2) {$\beta_2$};
		\node[circle](mp) at (-2, -2)  {{\Large $\oplus$}};
		\draw[->] (mp) -- (s1) ;
		\draw[->](-4, -2)-- (mp) ;
		\draw[->] (s1) -- (s2); 
		\draw[->] (s2) -- (4, -2)   -- (4, 0) -- (fb);
		\draw[->] (fb)-- (-2, 0) -- (mp); 
		\draw[->, red, ultra thick] (-3, -2.6) -- (-3, -2.1) node [pos=0, below] {$A_1$} -- (0.9, -2.1) node [pos=0.5, above] {$A'_1$}node [pos=1, above] {$D_1$} -- (0.9, -2.6);
		\draw[->, blue, ultra thick] (1.3, -2.6) -- (1.3, -2.1) node [pos=0, below] {$A_2$}  -- (4.7, -2.1) node [pos=0.9, above] {$D_2$}-- (4.7, -2.6);
	\end{tikzpicture}
	\caption{Counter-example for the analysis of a feedback control network. }
	\label{fig:net-cex}
\end{figure}
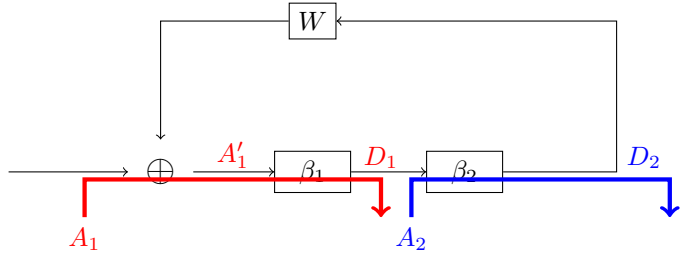

Considering this simplified system, we have the relations 
$$\begin{array}{l}
	A'_1 = A_1\land D_2+W\\
	D_2\geq A_2\conv \beta_2\\
	D_1 \geq A'_1\conv \beta_1.
\end{array}
$$
There is no feedback relation between $A'_1$ and the departure process $D_1$, so that this system can be solved easily (this is a system with three equations and 3 unknown variables, so in the classical algebra, one would not be able to give smallest solution). Worse, (a) one can create instability; (b) the backlog inside each server can be more than $W$, so the feedback control is inefficient. 

\begin{enumerate}
	\item[(a)] {\bf The system can be unstable} 
	Assume that $A_i$ is $\alpha_i:t\mapsto b_i+r_it$-constrained, $i\in\{1, 2\}$, and that $\beta_j = R_j(t-T_j)_+$, $j\in\{1, 2\}$. We can set and have
	$$
	\begin{array}{l}
		A_1 = \alpha_1 \implies  A'_1\leq A_1\leq \alpha_1 \\
		A_2 = \alpha_2 \implies D_2 \leq \alpha_2\implies	A'_1 \leq \alpha_2+W 	\\
		\implies A'_1 \leq \alpha_1\land (\alpha_2+W) 
		\implies A_1-A'_1 \geq \alpha_1-(\alpha_1 \land (\alpha_2+W))\\
		\implies \forall t\geq 0,~(A_1-A'_1)(t) \geq [(\alpha_1 -(\alpha_2+W))]_+(t)  = (r_1-r_2)\left[ \frac{b_1-W-b_2}{r_1-r_2}-t\right]_+.
	\end{array}$$
	If $r_1>r_2$, then $(A_1-A'_1)(t)$ grows to infinity and data can cross the controller at rate $r_2$. This creates an instability of the system, as the backlog of data  before the control-loop grows to infinity, hence delays also grows to infinity, and the stability hypothesis would be the unusual  $r_1\leq r_2$, and even more strict: $r_1\leq \underline{r}_2$ (where $\underline{r}_2$ is the minimal arrival rate of flow 2). 
	
	Although ~\cite{HCS24} considers  minimal arrival curves to guarantee finite delays, providing that an infinite amount of data is transmitted,  this cannot apply here. The stability requires to lower bound the maximal arrival rate, and could lead to unsolvable cases: if $R_2<r_1$, then this imposes that $R_2< r_1\leq r_2$ hence  the instability of the second server. 
	
	\item[(b)]{\bf The system is not controlled.} Assume that the system is stable (data of flow 2 are served faster in average than data of flow 1 arrive). The aim of the feedback controlled is to maintain the buffer occupancy below $W$ in each server. Assume that, without feedback control,  the backlog upper bound for the first server is $B_{\max, 1}>W$ and for the second server is $B_{\max, 2}>W$. First, there is no control on  $A_2$, that can arrive according to its arrival curve $\alpha_2$. Then, one can obtain a buffer of size $$b_{\max} = \sup_{t}(\alpha_2(t)- \beta_2(t) = B_{\max, 2}>W.$$  
	If moreover, and by stability assumption,  the arrival rate in server 2 is higher than in server 1, then, the feedback control is not limiting, and the quantity $D_2+W -A_1$ can be arbitrary large, so that data can enter the system according to $\alpha_1$, and the backlog $B_{\max, 1}$ can be reached. 
	
	%after some time, an arbitrary large amount of data can arrive, for example 
	%Also assume now that $\alpha_2$ is large enough  so that it is not limiting, and the arrival process $A'_1 = A_1$.  If 
	%$b_1+r_1T_2>W$, then the buffer of server 1 can have a buffer occupancy larger than $W$. 
	%Similar, if $b_2+r_2T_2> W$, the backlog of server 2 can be larger than $W$. 

%	\item [(c)] {\bf combination of the two.} Assume that $T_2$ is large enough. It should be possible to build a trajectory that combines all these difficulties: start by sending few data for flow 2: a backlog is created at the synchronization. Then start sending data of flow 2 so that it is saturated (with regards to its arrival curve), at a start of backlogged period. After $T_2$, the backlog created is $b_2+ r_2T_2$, that we can assume $>W$ data of flow 2 is sent at rate $R_2$, during  an interval of time $\frac{b_2+r_2T_2}{R_2-r_2}$, then at rate $r_2$ (this is not the minimum, but a possible behavior of the server. At the synchronization point, the amount of data that arrive at server 1 is data arriving at rate $R_2$ (we assumed we have created a large enough backlog), so data create a backlog $T_1 * R_2 \frac{b_2+r_2T_2}{R_2-r_2}> T_1 $ ... to be continued
\end{enumerate}

In order to avoid such issues in the analysis, we can enforce the following rule for complex feed-back structures: 
\begin{itemize}
	\item [(H)] The flow control is made for exactly the same flows that enter the system. 
\end{itemize}

Solving the issue regarding the example of ~\cite{HWS25} would be performing the control only for flows 1 an 2 (that crosse the whole system), and discard flows 4 and 5 from this analysis. But, if some synchronizations must be discarded, then the analysis is not valid anymore, since equivalent servers are used. 

In the next section, we provide means to perform a multi-flow analysis in simpler cases than the general statement of~\cite{HWS25}. The key idea, is not replacing each server by an equivalent server, but adding the throttle  servers while ensuring  that $(H)$ is satisfied. The interesting fact is since the service curve of the window control is sub-additive, and provided the other servers offer strict service curves, PMOO analysis is possible. 

\subsection{Some cases where the analysis is tractable}
In this section we show some examples where a safe analysisia possible. To avoid too much formalism, we only illustrate this section with examples, and believe that the more general results can be easily derived. 
\subsubsection{Each window controls only one flow}
When each window  controls one flow, the analysis is possible because the throttle service curve can be computed for each individual flow. This results in  a sink-tree open-loop network. 
\begin{example}
	Consider the network of Figure~\ref{fig:oneperwind}. Flow 1 in controlled by window $W_{1}$, flow 2 by window $W_{2}$ and flow 3 by window $W_{3}$. The open-loop version is  a tree, where the sources are thee throttle servers, as shown in Figure~\ref{fig:opw-unfold}. Let $\beta_i^{(j)}$ be the per-flow service curve for server $j$ and flow $i$: we have the equations:
	$$
	\left\{
	\begin{array}{ll}
		A'_1 = A_1\land D_1\conv\varphi_{W_{1}}\\
		D_1\geq A'_1\conv \beta_1^{(1)}\conv\beta_1^{(2)}\conv\beta_1^{(3)},
	\end{array}\right.
	\quad \text{so} \quad A'_1\geq A_1\conv (\beta_1^{(1)}\conv\beta_1^{(2)}\conv\beta_1^{(3)}\conv\varphi_{W_{1}})^*,
	$$
	$$
	\left\{
	\begin{array}{ll}
		A'_2 = A_2\land D_2\conv\varphi_{W_{2}}\\
		D_2\geq A'_2\conv \beta_2^{(1)}\conv\beta_2^{(2)},
	\end{array}\right.
	\quad \text{so} \quad A'_2\geq A_2\conv (\beta_2^{(1)}\conv\beta_2^{(2)}\conv\varphi_{W_{2}})^*,
	$$
	and
	$$ 
	\left\{
	\begin{array}{ll}
		A'_ = A_3\land D_3\conv\varphi_{W_{3}}\\
		D_3\geq A'_3 \conv \beta_3^{(2)}\conv\beta_3^{(3)},
	\end{array}\right.
	\quad \text{so} \quad A'_3\geq A_3\conv (\beta_3^{(2)}\conv\beta_3^{(3)}\conv\varphi_{W_{3}})^*.
	$$
	
	From Example~\ref{ex:output}, in case of stability, there the throttle server does not induce any increase in the arrival curve at its output, so one can compute the per-flow service curves the as usual, using the service policy of interest, with the strict service curve assumption if needed. In this case, the throttle server is only crossed by one flow, it is already a per-flow service curve. In the case of blind multiplexing, the PMOO analysis is possible (with strict service curves). 
		
		\begin{figure}[htbp]
			\centering
			\begin{tikzpicture}
				\node[rectangle, draw, minimum width=0.5cm](fb1) at (0, 0) {$W_{2}$};
				\node[rectangle, draw, minimum width=0.5cm](fb2) at (3,  -1) {$W_{3}$};
				\node[rectangle, draw, minimum width=0.5cm](fb3) at (3,  -3) {$W_{1}$};
				\node[rectangle, draw, minimum width=1cm](s1) at (0, -2) {$\beta_1$};
				\node[rectangle, draw, minimum width=1cm](s2) at (3, -2) {$\beta_2$};
				\node[rectangle, draw, minimum width=1cm](s3) at (5, -2) {$\beta_3$};
				\node[circle](mp1) at (-1.5, -2)  {{\huge $\oplus$}};
				\node[circle](mp2) at (1.5, -2)  {{\huge $\oplus$}};
				\node[circle](mp3) at (-2.5, -2)  {{\huge $\oplus$}};
				\draw (-2.5, -2) -- (s1) -- (s2) -- (s3) ;

				\draw[->] (s2) -- (4, -2)   -- (4, 0) -- (fb1) -- (-1.5, 0) -- (-1.5,-1.75) ;
				\draw[->] (s3) -- (5.7, -2)   -- (5.7, -1) -- (fb2) -- (1.5, -1) -- (1.5,-1.75) ;
				\draw[->] (s3) -- (5.7, -2)   -- (5.7, -3) -- (fb3) -- (-2.5, -3) -- (-2.5,-2.2) ;
				\draw[red, ultra thick, ->] (-3.5, -2.1) -- (-2, -2.1) node[pos=0, below] {$1$} node[pos=0.2, below, black] {$A$}  -- (-1.5, -2.5) -- (-1, -2.1) -- (1, -2.1) node[pos=0.1, below, black] {$A'$}  node[pos=0.9, below, black] {$B$}  -- (1.5, -2.5) -- (2, -2.1) -- (6.5, -2.1)node[pos=0.05, below, black] {$B'$} node[pos=0.5, below, black] {$C$} node[pos=0.9, below, black] {$D$};
				\draw[blue, ultra thick, ->] (1, -1.9) -- (6.5, -1.9) node[pos=0, above  right = -0.05cm ] {$3$}  ;
				\draw[olive, ultra thick, ->] (-2.5, -1.5) --(-2, -1.5)node[pos=0, above] {$2$} --(-2, -1.9) -- (1, -1.8) -- (1.5, -1.3) -- (2, -1.8) -- (4.3, -1.8)   ;
			\end{tikzpicture}
			\caption{Tandem network with a complex feedback structure, where each window controls one flow. The letters $A, A'$... indicate the notation for the cumulative processes (with the index corresponding to the flow).}
			\label{fig:oneperwind}
		\end{figure}
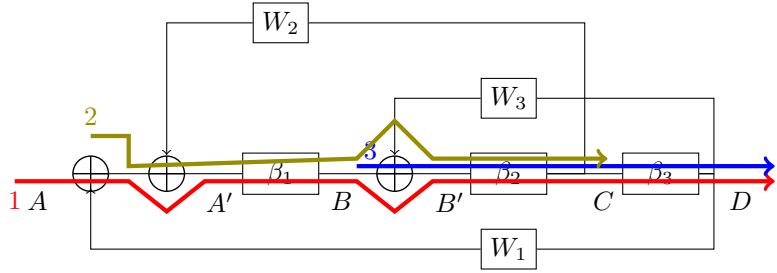
		
		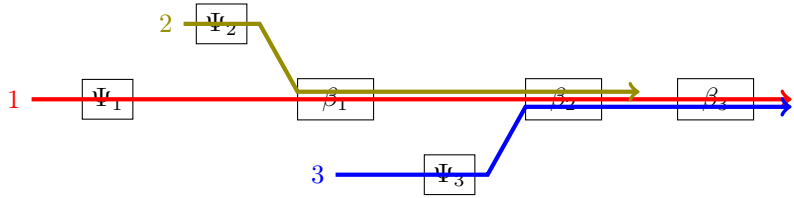
\begin{figure}[htbp]
			\centering
			\begin{tikzpicture}
				\node[rectangle, draw, minimum width=0.5cm](fb1) at (0, 0) {$\Psi_1$};
				\node[rectangle, draw, minimum width=0.5cm](fb2) at (1.5,  1) {$\Psi_2$};
				\node[rectangle, draw, minimum width=0.5cm](fb3) at (4.5,  -1) {$\Psi_3$};
				\node[rectangle, draw, minimum width=1cm](s1) at (3, 0) {$\beta_1$};
				\node[rectangle, draw, minimum width=1cm](s2) at (6, 0) {$\beta_2$};
				\node[rectangle, draw, minimum width=1cm](s3) at (8, 0) {$\beta_3$};

				\draw[red, ultra thick, ->] (-1, 0) -- (9, 0) node[pos=0, left] {1}   ;
				\draw[olive, ultra thick, ->] ((1, 1) -- (2, 1)node[pos=0, left] {2}  -- (2.5, 0.1) -- (7, 0.1);
				\draw[blue, ultra thick, ->] ((3, -1) -- (5, -1)node[pos=0, left] {3}  -- (5.5, -0.1) -- (9, -0.1);
			\end{tikzpicture}
			\caption{Open-loop version of the network of Figure~\ref{fig:cpx-ex}. }
			\label{fig:opw-unfold}
		\end{figure}
		
\end{example}

To summarize, the main steps of the analysis are:

\begin{enumerate}
	\item In case of stability, then the throttle server does not modify the arrival curves;
	\item Compute all the per-flow end-to-end service curves for each flow, as if there were no window-flow control (in the original tandem network);
	\item Compute the throttle service curve for each flow;
	\item Check the stability of the system. If stable, then the performance bound is safe, otherwise, there might be a stability issue. 
\end{enumerate}

\subsubsection*{Nested tandem networks}

The second case is when windows and flows are nested: for each pair of flows $i$ and $j$, either the two paths are disjoint ($f_j>\ell_i$ or $f_i> \ell_j$), or one path is a sub-path of the other ($f_i\leq f_j\leq \ell_j\leq \ell_i$ or $f_j\leq f_i\leq \ell_i\leq \ell_j$), and similar for windows. 
One can modify the previous example with flow 2 crossing server 2 only to obtain a nested tandem, and the throttle servers are crossed by several flows now. Note that window $W_{2}$ controls flows 1, 2 and 3, whereas window $W_{3}$ only controls flows 1 and 3, and window $W_1$ controls flow 1. The system is depicted on Figure~\ref{fig:nested}.

We have the inequalities, starting from the more nested window: first for window 3

\begin{multline*}
	\left.
	\begin{array}{ll}
		B''_1+B''_2+B''_3 = (B'_1 + B'_2+B'_3) \land (C_1 + C_2+C_3)\conv\varphi_{W_{3}}\\
		(C_1+C_2+C_3)\geq (B''_1+B''_2+B''_3) \conv \beta_2\\
	\end{array}
	\right.\\
	\quad\text{so}\quad B''_1+B''_2+B''_3 \geq (B'_1+ B'_2+B'_3)\conv (\beta_2\conv\varphi_{W_{3}})^*,
\end{multline*}
so $\Psi_3 = (\beta_2\conv\varphi_{W_{3}})^*$,
and for window 2 (here we use blind multiplexing and assume servers offer strict service curves), 

$$\left.
\begin{array}{ll}
	B'_1+B'_3 = (B_1+B_3) \land (D_1 +C_3)\conv\varphi_{W_{3}}\\
	(D_1+D_3)\geq (B'_1+B'_3) \conv (\beta_2\conv \Psi_3-\alpha_2)_+\conv\beta_3\\
\end{array}
\right.
\quad\text{so}\quad B'_1+B'_3 \geq (B_1+B_3)\conv (\beta_2\conv \Psi_3-\alpha_2)_+\conv\beta_3\conv\varphi_{W_{3}})^*,
$$
so $\Psi_2 = (B_1+B_3)\conv (\beta_2\conv \Psi_3-\alpha_2)_+\conv\beta_3\conv\varphi_{W_{3}})^*$
Here we use the partial analysis for window 3 and its throttle server: we can used the PMOO analysis to compute a guarantee on the service curve  for the concatenation of servers $\beta_2$ and $\Psi_3$. 
Finally, using the same tools, we can analyze window 1:

\begin{multline*}\left.
\begin{array}{ll}
	A_1 = A'_1\land D_1\conv\varphi_{W_{1}}\\
	D_1\geq A'_1\conv \beta_1 \conv \Psi_2\conv (\Psi_2\conv \beta_2)-\alpha_2)_+ \conv\beta_3 - \alpha_3)_+\\
\end{array}\right.\\
\quad\text{so}\quad A'_1 \geq A_1\conv (\beta_1 \conv \Psi_2\conv (\Psi_2\conv \beta_2)-\alpha_2)_+ \conv\beta_3 - \alpha_3)_+\conv\varphi_{W_{1}})^*
\end{multline*}

 The equivalent open-loop network is depicted in Figure~\ref{fig:nested-open}.
If the system is stable, then performance can be computed as previously. 
The stability, as in the previous case can be checked afterwards (and even after each window analysis, if the stability condition holds for each window considering the service curves computed. Dimensioning the window to ensure stability (in blind multiplexing) can be done in succession for each flow.

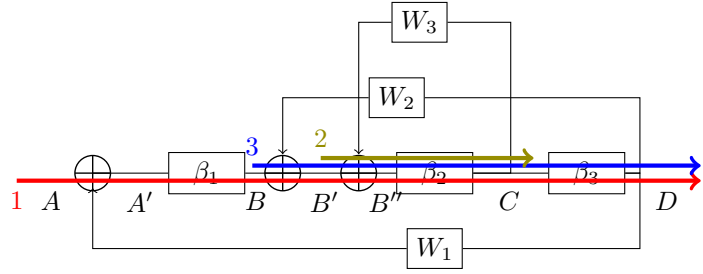
\begin{figure}[htbp]
	\centering
	\begin{tikzpicture}
		\node[rectangle, draw, minimum width=0.5cm](fb1) at (2.8, 0) {$W_{3}$};
		\node[rectangle, draw, minimum width=0.5cm](fb2) at (2.5,  -1) {$W_{2}$};
		\node[rectangle, draw, minimum width=0.5cm](fb3) at (3,  -3) {$W_{1}$};
		\node[rectangle, draw, minimum width=1cm](s1) at (0, -2) {$\beta_1$};
		\node[rectangle, draw, minimum width=1cm](s2) at (3, -2) {$\beta_2$};
		\node[rectangle, draw, minimum width=1cm](s3) at (5, -2) {$\beta_3$};
		\node[circle](mp1) at (-1.5, -2)  {{\huge $\oplus$}};
		\node[circle](mp2) at (1, -2)  {{\huge $\oplus$}};
		\node[circle](mp3) at (2, -2)  {{\huge $\oplus$}};
		\draw (-1.5, -2) -- (s1) -- (s2) -- (s3) ;

		\draw[->] (s2) -- (4, -2)   -- (4, 0) -- (fb1) -- (2, 0) -- (2,-1.75) ;
		\draw[->] (s3) -- (5.7, -2)   -- (5.7, -1) -- (fb2) -- (1, -1) -- (1,-1.75) ;
		\draw[->] (s3) -- (5.7, -2)   -- (5.7, -3) -- (fb3) -- (-1.5, -3) -- (-1.5,-2.2) ;
		\draw[red, ultra thick, ->] (-2.5, -2.1) -- (6.5, -2.1) node[pos=0, below] {$1$} node[pos=0.05, below, black] {$A$}
		node[pos=0.18, below, black] {$A'$} node[pos=0.35, below, black] {$B$} node[pos=0.45, below, black] {$B'$}node[pos=0.54, below, black] {$B''$} node[pos=0.72, below, black] {$C$} node[pos=0.95, below, black] {$D$};
		\draw[blue, ultra thick, ->] (0.6, -1.9) -- (6.5, -1.9) node[pos=0, above] {$3$}  ;
		\draw[olive, ultra thick, ->] (1.5, -1.8) -- (4.3, -1.8) node[pos=0, above] {$2$}  ;
	\end{tikzpicture}
	\caption{Nested tandem network with a nested feedback structure.  }
	\label{fig:nested}
\end{figure}

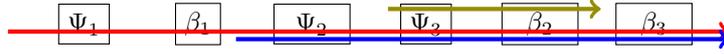
\begin{figure}[htbp]
	\centering
	\begin{tikzpicture}
		\node[rectangle, draw, minimum width=0.5cm](fb1) at (0, 0) {$\Psi_1$};
		\node[rectangle, draw, minimum width=0.5cm](fb2) at (1.5,  0) {$\beta_1$};
		\node[rectangle, draw, minimum width=0.5cm](fb3) at (4.5,  0) {$\Psi_3$};
		\node[rectangle, draw, minimum width=1cm](s1) at (3, 0) {$\Psi_2$};
		\node[rectangle, draw, minimum width=1cm](s2) at (6, 0) {$\beta_2$};
		\node[rectangle, draw, minimum width=1cm](s3) at (7.5, 0) {$\beta_3$};

		\draw[red, ultra thick, ->] (-1, -0.1) -- (8.5, -0.1) ;
		\draw[olive, ultra thick, ->] (4, 0.2) -- (6.8, 0.2) ;
		\draw[blue, ultra thick, ->] (2, -0.2) -- (8.5, -0.2) ;

	\end{tikzpicture}
	\caption{Open-loop version of the network of Figure~\ref{fig:nested}. }
	\label{fig:nested-open}
\end{figure}

\subsection{Some hints for the analysis of more complex feed-back controls}
We now  turn to a case  where the analysis seems untractable: we consider the topology of Figure ~\ref{fig:interleave} with three servers, three flows and two servers. The flows are intricate.  According to $(H)$, the window $W_{2}$ can only control flows 1 an 3, and window $W_{1}$ only flows 1 and 2.

\begin{figure}[htbp]
	\centering
	\begin{tikzpicture}
		\node[rectangle, draw, minimum width=0.5cm](fb1) at (0, 0) {$W_{1}$};
		\node[rectangle, draw, minimum width=0.5cm](fb2) at (3,  -1) {$W_{2}$};
		\node[rectangle, draw, minimum width=1cm](s1) at (0, -2) {$\beta_1$};
		\node[rectangle, draw, minimum width=1cm](s2) at (3, -2) {$\beta_2$};
		\node[rectangle, draw, minimum width=1cm](s3) at (5, -2) {$\beta_3$};
		\node[circle](mp1) at (-1.5, -2)  {{\huge $\oplus$}};
		\node[circle](mp2) at (1.5, -2)  {{\huge $\oplus$}};
		\draw[->] (mp1) -- (s1) ;

		\draw[->] (s2) -- (4, -2)   -- (4, 0) -- (fb1) -- (-1.5, 0) -- (-1.5,-1.8) ;
		\draw[->] (s3) -- (5.7, -2)   -- (5.7, -1) -- (fb2) -- (1.5, -1) -- (1.5,-1.8) ;
		
		\draw[->, red, ultra thick] (-2, -1.4) -- (-2, -1.9)  node[pos=0, left] {2} -- (1, -1.9)-- (1.5, -1.5)--(2, -1.9) -- (4.2, -1.9) -- (4.2, -1.4);
		\draw[->, blue, ultra thick] (1.3, -2.6) -- (1.3, -2.1) node[pos=0, left] {3} -- (6, -2.1) -- (6., -2.6);
		\draw[olive, ultra thick, ->] (-2.5, -2) -- (6.5, -2)node[pos=0, left] {1} node[pos=0.05, below, black] {$A$}
		node[pos=0.17, below, black] {$A'$}   node[pos=0.38, below, black] {$B$} node[pos=0.5, below, black] {$B'$}   node[pos=0.73, below, black] {$C$} node[pos=0.97, below, black] {$D$};
	\end{tikzpicture}
	
	\bigskip
	
	\begin{tikzpicture}
		\node[rectangle, draw, minimum width=1cm](fb1) at (1.5, 0) {$\beta_1$};
		\node[rectangle, draw, minimum width=1cm](fb2) at (4.5,  0) {$\beta_2$};
		\node[rectangle, draw, minimum width=1cm](fb3) at (6,  0) {$\beta_3$};
		\node[rectangle, draw, minimum width=0.5cm](s1) at (0, 0) {$\Psi_ 1$};
		\node[rectangle, draw, minimum width=0.5cm](s2) at (3, 0) {$\Psi_2$};

		\draw[red, ultra thick, ->] (-1, -0.1) -- (7, -0.1)node[pos=0, left] {1}  ;
		\draw[olive, ultra thick, ->] (-1, 0.2) -- (2.5, 0.2) node[pos=0, left] {2} -- (3, 0.7) -- (3.5, 0.2) -- (5.3, 0.2) ;
		\draw[blue, ultra thick, ->] (2.3, -0.2) -- (7, -0.2) node[pos=0, below] {3} ;
	\end{tikzpicture}

	\caption{Interleaved network, with two feedback loops. The first loop controls flow 1 and 2, while the second loop controls flows 1 and 3, and (below) its open-loop version, that is not a tandem network.  }
	\label{fig:interleave}
\end{figure}
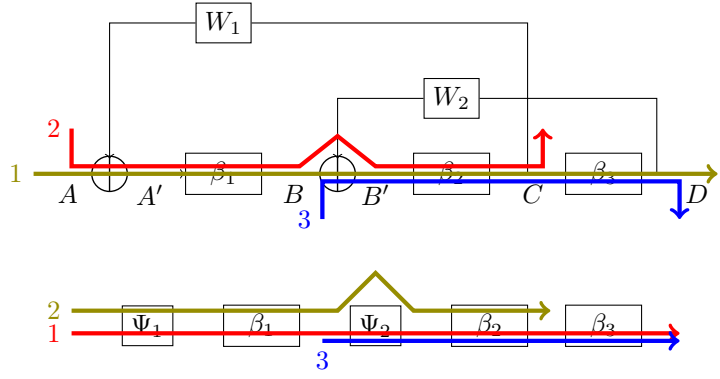

We can write the equalities and inequalities that must be satisfied:

$$
\begin{array}{ll}
	A'_1+A'_2 = (A_1+A_2)\land (C_1+C_2)\conv \varphi_{W_{1}}& B'_1+B'_3 = (B_1+B_3)\land (D_1+D_3)\conv \varphi_{W_{2}}\\
	C_1+C_2 +C_3\geq (B'_1+B_2+B'_3)\conv \beta_2 &  (D_1+D_3)\geq (C_1+C_3)\conv\beta_3\\
	B_1+B_2 \geq (A'_1+A'_2)\conv \beta_1 &
	%A'_1+A'_2 \geq (A_1+A_2)\conv ((\beta_1 \conv\beta_2)\conv \varphi_{W_{1, 2}})^*&
\end{array}
$$

One approach could be to first compute $\beta_2^{(1, 3)}$, the residual service curve of server 2 for flows 1 and 3. But for this, we need to compute the arrival curve for flow 2 at server 2, $\alpha_2^{(2)}$, and then we have, for the analysis of the second window: 
$$B'_1+B'_3\geq (B_1+B_3)\conv\big((\beta_2-\alpha_2^{(2)})_+\conv\beta_3\conv\varphi_{W_{2}}\big)^*.$$

 In return we need to be able to compute the throttle service curve fo $W_{1}$, and we need the residual service curve for server 2 and flows 1 and 2. Continuing the same idea would require first the knowledge of $\beta_2^{(1, 3)}$: we have a fixed point to compute, which is beyond the scope of this paper. 

Alternatively, we can assume fixed priorities: flow 2 has the highest priority, and flow 1 the lower, for both the synchronization (will be the throttle servers) and the server. Since flow 2 has the highest priority,the feedback-loop $W_{1}$ only sees flow 2, and one can compute the throttle server for flow 2: $(\beta_1\conv\beta_2\conv\varphi_{W_{1}})^*$. The arrival curve of flow 2 at server 2 is still $\alpha_2$, so the service curve for flows 1 and 3 is $(\beta_2-\alpha_2)_+$. The throttle  server $$\Psi_{2} = ((\beta_2-\alpha_2)_+\conv \beta_3\conv\varphi_{W_{2}})^*.$$

In this case, we face the problem of analyzing (and finding the ans-to-end service curve for two flows that do not follow the same path. Such an analysis has to been done in NC and is, again,  beyond the scope of this paper. 

If we restrict the window-control system so that the open-loop network remains a tandem network (or a tree, similar to the previous example), the only {\em authorized} feedback controls are represented by the two triples $(i, i, W_{i})$, $i\in\{1, 2, 3\}$ and $(1, 2, W_{4})$. Server 3 cannot be included in a window with other servers. This maximal set of window control is depicted in Figure~\ref{fig:max}, and deriving the open-loop system  is given by the equations, and the throttle servers:

\begin{align*}
	A'_1+A'_2 &= (A_1+A_2)\land (B_1+B_2)\conv \varphi_{W_{2}} \\
	 B_1+B_2 &\geq (A'_1+A'_2)\conv \beta_1\\
	B'_1+B'_2+B'_3 &= (B_1+B_2+B_3)\land (C_1+C_2+C_3)\conv \varphi_{W_{3}}\\ C_1+C_2+C_3&\geq (B'_1+B'_2+B'_3)\conv \beta_2\\
  C'_1+C'_3 &= (C_1+C_3)\land (D_1+D_3)\conv \varphi_{W_{4}} \\
   D_1+D_3&\geq (C'_1+C'_3)\conv \beta_3\\
	A'_1+A'_2& = (A_1+A_2)\land (C_1+C_2)\conv \varphi_{W_{1}},
\end{align*}
and the throttle service curves are:

\begin{align*}
	\Psi_4 &= (\beta_3\conv \varphi_{W_{4}})^*\\
	\Psi_3 &= (\beta_2\conv \varphi_{W_{3}})^*\\
	\Psi_2 &= (\beta_1\conv \varphi_{W_{2}})^*\\
	\Psi_1 &= (\beta_1\conv(\Psi_3\conv\beta_2-\alpha_3)_+\conv \varphi_{W_{1}})^*,
\end{align*}
with the blind multiplexing assumption. Here again, after this computation, stability has to be checked to ensure the validity of the result. 
Note that in this case, we need to compute the per-flow service curve for flow 1 and 2. Due to the tandem structure, this is removing flow 2 from servers $\Psi_3$ and $\beta_2$, hence the formula. Note that the use of the sub-additive closure operator ensures that $\Psi_1$ is sub-additive, and  the whole the open-loop system does not include per-flow service at all. There is a priori no simple way to compute residual service curves for sub-additive functions, but computations can be performed using the Nancy toolbox~\cite{ZS22}. 
\begin{figure}[htbp]
	\centering
	\begin{tikzpicture}
		\node[rectangle, draw, minimum width=0.5cm](fb1) at (0, -1) {$W_{2}$};
		\node[rectangle, draw, minimum width=0.5cm](fb2) at (3,  -1) {$W_{3}$};
		\node[rectangle, draw, minimum width=0.5cm](fb3) at (6, -1){$W_{4}$};
		\node[rectangle, draw, minimum width=0.5cm](fb4) at (1.5,  -3) {$W_{1}$};
		\node[rectangle, draw, minimum width=1cm](s1) at (0, -2) {$\beta_1$};
		\node[rectangle, draw, minimum width=1cm](s2) at (3, -2) {$\beta_2$};
		\node[rectangle, draw, minimum width=1cm](s3) at (6, -2) {$\beta_3$};
		\node[circle](mp1) at (-1.5, -2)  {{\huge $\oplus$}};
		\node[circle](mp2) at (1.5, -2)  {{\huge $\oplus$}};
		\node[circle](mp3) at (4.5, -2)  {{\huge $\oplus$}};
		\draw[->] (mp1) -- (s1) ;

		\draw[->] (1, -2) -- (1, -1) -- (fb1) -- (-1.5, -1) -- ((-1.5, -1.75); 
		\draw[->] (4, -2) -- (4, -1) -- (fb2) -- (1.5, -1) -- ((1.5, -1.75); 
		\draw[->] (7, -2) -- (7, -1) -- (fb3) -- (4.5, -1) -- ((4.5, -1.75); 
		\draw[->] (4, -2) -- (4, -3) -- (fb4) -- (-1.5, -3) -- ((-1.5, -2.25); 
		\draw[->, olive, ultra thick] (-2, -1.4) -- (-2, -1.9)  node[pos=0, left] {2}  -- (4.2, -1.9) -- (4.2, -1.4);
		\draw[->, blue, ultra thick] (1.2, -2.6) -- (1.2, -2.1) node[pos=0, left] {3} -- (7.5, -2.1) -- (7.5, -2.6);
		\draw[red, ultra thick, ->] (-2.5, -2) -- (7.5, -2)node[pos=0, left] {1} node[pos=0.05, below, black] {$A$}
		node[pos=0.15, below, black] {$A'$}   node[pos=0.35, below, black] {$B$} node[pos=0.45, below, black] {$B'$}   node[pos=0.63, below, black] {$C$}node[pos=0.75, below, black] {$C'$} node[pos=0.93, below, black] {$D$};
	\end{tikzpicture}
	
	\medskip
		\begin{tikzpicture}
			\node[rectangle, draw, minimum width=0.5cm](fb1) at (0, 0) {$\Psi_1$};
			\node[rectangle, draw, minimum width=0.5cm](fb2) at (1.5, 0) {$\Psi_2$};
			\node[rectangle, draw, minimum width=0.5cm](fb3) at (4.5,  0) {$\Psi_3$};
			\node[rectangle, draw, minimum width=0.5cm](fb4) at (7.5,  0) {$\Psi_4$};
			\node[rectangle, draw, minimum width=1cm](s1) at (3, 0) {$\beta_1$};
			\node[rectangle, draw, minimum width=1cm](s2) at (6, 0) {$\beta_2$};
			\node[rectangle, draw, minimum width=1cm](s3) at (9, 0) {$\beta_3$};

			\draw[red, ultra thick, ->] (-1, -0.1) -- (10, -0.1) ;
			\draw[olive, ultra thick, ->] (-1, 0.2) -- (7, 0.2) ;
			\draw[blue, ultra thick, ->] (4, -0.2) -- (10, -0.2) ;

		\end{tikzpicture}

	\caption{Maximal set of windows for a possible analysis with blind multiplexing, and, below the equivalent open-loop network is a tandem network.  }
	\label{fig:max}
\end{figure}
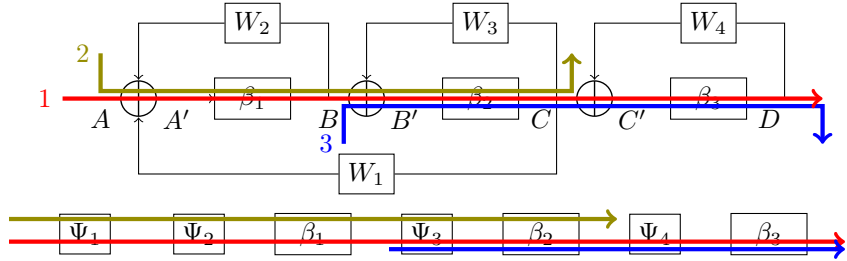

These latter examples allow to outline a more general case for a tandem network with general feedback structures that can be analyze. They must obey to several principles, in addition to (H): 
\begin{enumerate}
	\item {\bf Flow separation:} The flows can be separated in several groups do that the windows do not interfere in-between groups, the extreme case being one flow per window. 
	\item {\bf Control Nesting:} Inside each group of flows, and similar to the last example, for each group of flow, one can define the maximal set of windows, that make the analysis possible. Our intuition is that these controls that should be nested, and choosing a subset of this maximal set makes the analysis possible
\end{enumerate}
Of course, this last remark remains unproved, and might require more attention.

\section {Conclusion}

In this paper, we demonstrate that sub-additive functions are a class for which non-negative, (min,plus), per flow service curves can be derived without additional hypothesis. This has an impact on the study of tandem networks with complex feedback-control structures in a multi-flow setting. We also derive a PMOO analysis of tandem networks with a mix of sub-additive service curves, transmission delay and strict service curves. Among the other implications of this work, the use-cases presented in~\cite{HCS24} have an analysis that does not require service curves with negative values. We also demonstrate that complex feed-back structures can have stability issues, and can be more complex than expected in~\cite{HWS25} to analyze. As a conclusion, the use of service curves with negative values is still to be demonstrated.

\end{document}